\documentclass[aps,pra,reprint,superscriptaddress,nofootinbib]{revtex4-2}

\usepackage{amsmath,amssymb,mathtools,bm,amsthm}
\usepackage{graphicx}
\usepackage{booktabs}
\usepackage{microtype}
\usepackage{tikz}
\usepackage[hidelinks]{hyperref}

\newtheorem{theorem}{Theorem}
\newtheorem{proposition}[theorem]{Proposition}
\newtheorem{lemma}[theorem]{Lemma}
\newtheorem{corollary}[theorem]{Corollary}

\newcommand{\C}{\mathbb C}
\newcommand{\R}{\mathbb R}
\newcommand{\Tr}{\operatorname{Tr}}
\newcommand{\rank}{\operatorname{rank}}
\newcommand{\id}{\operatorname{id}}
\newcommand{\End}{\operatorname{End}}
\newcommand{\Sym}{\operatorname{Sym}}
\newcommand{\proj}[1]{|#1\rangle\!\langle#1|}
\newcommand{\ket}[1]{|#1\rangle}
\newcommand{\bra}[1]{\langle#1|}
\newcommand{\Oqd}{\mathcal O_q^{(d)}}
\newcommand{\Nd}{N_d}
\newcommand{\Ad}{A_d}
\newcommand{\Cd}{C_d}

\begin{document}

\title{Asymmetric quantum cloning on orthogonal orbits and four-mode fermionic states}

\author{Piotr {\'C}wikli{\'n}ski}
\affiliation{International Centre for Theory of Quantum Technologies, University of Gda{\'n}sk, ul. prof. Marii Janion 7, 80-309 Gda\'nsk, Poland}

\author{Micha{\l} Studzi{\'n}ski}
\affiliation{International Centre for Theory of Quantum Technologies, University of Gda{\'n}sk, ul. prof. Marii Janion 7, 80-309 Gda\'nsk, Poland}

\date{\today}

\begin{abstract}
In this work, we study asymmetric $(1\to2)$ quantum cloning for pure
states belonging to orbits generated by the orthogonal group. For every
dimension $(d\geq3)$, we determine the full region of achievable pairs of
average single-copy fidelities. Using orthogonal covariance and the Brauer
algebra, we reduce the optimization over all CPTP maps to a
finite-dimensional spectral problem and construct explicit optimal cloning
channels. We then apply the general results to four-mode fermionic states
in the even-parity sector. Using triality, we identify the fixed-concurrence fermionic families with
the orthogonal orbits considered in the first part of the paper. In particular, for pure Gaussian states we
show that the symmetric locally optimal cloner is unique. We also compare
local and global cloning and find that the channel optimal for the joint
two-copy fidelity is different from the one optimal for the single-copy
fidelities. Finally, we show that the locally optimal cloner cannot be
implemented using only operations which preserve fermionic Gaussian states.
\end{abstract}


\maketitle

\section{Introduction}
\label{sec:introduction}

The no-cloning theorem forbids perfect copying of an arbitrary unknown
quantum state \cite{WoottersZurek1982,Dieks1982}, but approximate copying is
possible.  Approximate cloners are assessed by local fidelities, namely the
overlaps of the reduced outputs with the unknown pure input.  For symmetric
cloning, the figure of merit is the common single-clone fidelity
\cite{BuzekHillery1996,GisinMassar1997,Werner1998,KeylWerner1999}, whereas
asymmetric cloning is described by the region of achievable fidelity pairs
\cite{CerfAsymmetric2000,IblisdirEtAl2005,KayRamanathanKaszlikowski2013,
Hashagen2017,NechitaPellegriniRochette2021,NechitaPellegriniRochette2023};
see also Refs.~\cite{Scarani2005,Fan2014}.  The Pareto boundary gives the optimal tradeoff between the fidelities of
the two copies.  Universal cloners
assume no information about the input beyond its Hilbert-space dimension
\cite{CwikHorStud2012,StudCwikHorMoz2014}.  If the input states are restricted to a group orbit, the optimal cloning
region can be different from the universal one.
Restricted-ensemble cloning has been studied for spin-coherent states,
bosonic coherent states, displaced thermal states, and group-covariant
quantum benchmarks
\cite{DemkowiczKusWodkiewicz2004,Braunstein2001,GutaMatsumoto2006,
YangChiribellaAdesso2014}.  In this work, we determine the complete \(1\mathbin{\to}2\) asymmetric
local-cloning region for pure-state orbits of the orthogonal group.

Our main application concerns four-mode fermionic states.  Fermionic
linear optics (FLO), generated by quadratic dynamics, preserves Gaussian
states and can be efficiently described in terms of covariance matrices
\cite{Terhal2002,Bravyi2005,JozsaMiyake2008,BravyiKonig2012}.  Several
computational models based on FLO can be efficiently simulated classically.
To obtain universality one has to add a non-Gaussian resource, for example
a suitable non-Gaussian ancilla or an interaction beyond the quadratic
level
\cite{BravyiKitaev2002,Bravyi2006,deMeloCwikTerhal2013,HebenstreitEtAl2019,
OBrienRozekAkhmerov2018,SierantStornatiTurkeshi2026}.  For one, two, and
three modes every pure state in a fixed parity sector is Gaussian.  The
situation changes for four modes, where Gaussian states form a proper
subset of all pure states and fermionic concurrence labels the pure-state
FLO orbits \cite{deMeloCwikTerhal2013,OszmaniecGuttKus2014}.  We therefore
study cloning on the fixed-concurrence orbits, including the Gaussian one.

To connect the fermionic problem with orthogonal-group cloning, we use
triality.  In the even-parity sector, the positive half-spin representation
of \(\mathrm{Spin}(8)\) can be related to the vector representation of
\(SO(8)\) \cite{Chevalley1954}.  With the triality intertwiner chosen below,
the squared fermionic concurrence becomes
\(q=|\psi^T\psi|^2\), while the pure Gaussian states correspond to the
orbit \(q=0\) \cite{OszmaniecGuttKus2014}.  More generally, for every
\(q\in[0,1]\), states with fixed \(q\) form a single projective
\(O(d)\) orbit in \(\mathbb{CP}^{d-1}\).  Thus the same cloning problem
covers the complex quadric at \(q=0\), the real-state ensemble at \(q=1\),
and all intermediate orbits.  We solve this problem for arbitrary
\(d\geq3\); the four-mode fermionic case is obtained by setting \(d=8\).

There is one special value of the orbit parameter,
\begin{equation}
 q_{\rm H}=\frac{2}{d+1}.
\end{equation}
For this value the second moment of the orbit is equal to the Haar second
moment, so the orbit is a continuous complex projective \(2\)-design and
the corresponding local-cloning problem reduces to universal cloning.
At the other endpoint, \(q=1\), we recover the real-state ensemble
\cite{NavezCerf2003,ZhangYe2007}.  We calculate the support function of
the complete asymmetric fidelity region for arbitrary real output weights
and construct explicit optimal covariant channels.  We also find that the
symmetric fidelity has a unique minimum at \(q=q_{\rm H}\).  In particular,
for \(d=8\) this minimum is reached neither for Gaussian states
(\(q=0\)) nor for states of maximal concurrence (\(q=1\)).  Thus, in the
four-mode case, the optimal symmetric cloning fidelity is not monotone in
fermionic concurrence.

To solve the optimization, we first use orthogonal covariance.  After
twirling, the Choi operator belongs to the commutant of the diagonal
\(O(d)\) action, which is described by the Brauer algebra
\cite{DarianLoPresti2001,Brauer1937,DotyHu2009}.  This allows us to write
the optimization in terms of permutation and contraction operators and to
reduce the original problem to matrices of small dimension.  For every
nonzero nonnegative supporting direction we then construct explicitly an
optimal covariant channel of Choi rank \(d\).

Finally, for the quadric orbit \(q=0\), we compare two different cloning
criteria.  Besides the local single-copy fidelities, we consider the global
fidelity with the ideal state \(\rho_\psi^{\otimes2}\).  The globally
optimal channel follows from the coherent-state construction of
Ref.~\cite{ChiribellaYang2014}, and we also give a direct derivation.
Interestingly, this channel is different from the symmetric locally
optimal cloner.  We prove that the latter is unique among all CPTP maps and
has Choi rank \(d\).  For \(d=8\), this gives an exact comparison between
local and global cloning of pure four-mode Gaussian states, and the global
value agrees with the coherent-orbit formula of
Ref.~\cite{Herasymenko2026}.  We stress that here we consider the
finite-copy \(1\to2\) problem, which is different from the many-input-copy
regime studied in Refs.~\cite{FanizzaEtAl2026,JeonSohnOh2026}.

Section~\ref{sec:orthogonal-problem} formulates the orbit problem,
Secs.~\ref{sec:region} and \ref{sec:channels} solve it and construct optimal
channels, Sec.~\ref{sec:fermions} gives the fermionic interpretation, and
Sec.~\ref{sec:null-local-global} compares local and global cloning.  Technical
calculations and the uniqueness proof are collected in the Appendices.

\section{Orthogonal-orbit cloning problem}
\label{sec:orthogonal-problem}

\subsection{Orbit geometry}

Let \(\mathcal H=\C^d\), \(d\geq3\), with a fixed orthonormal basis
\(\mathcal B_{\rm O}:=\{\ket{j}\}_{j=1}^d\).  We write \(O(d)\) for the subgroup
of \(U(d)\) whose matrices in \(\mathcal B_{\rm O}\) are real orthogonal;
\(SO(d)\) is its determinant-one subgroup.
Transposition and complex conjugation are taken in this basis.  In
Sec.~\ref{sec:fermions}, where \(\mathcal B_{\rm O}\) is compared with the
fermionic occupation-number basis, we temporarily write \(\ket{j}_{\rm O}\)
for the same vectors.  For a normalized vector \(\psi\in\mathcal H\), define
\begin{equation}
 q(\psi)=|\psi^T\psi|^2\in[0,1]
 \label{eq:orbit-parameter}
\end{equation}
and denote the corresponding projective level set by
\begin{equation}
 \Oqd=\{[\psi]\in\mathbb{CP}^{d-1}:|\psi^T\psi|^2=q\}.
 \label{eq:Oqd}
\end{equation}
Here \(\mathbb{CP}^{d-1}\) is the space of complex one-dimensional
subspaces (rays) in \(\mathcal H\), and \([\psi]\) is the ray containing
\(\psi\).  In Eq.~\eqref{eq:Oqd}, \(\psi\) is any normalized
representative of \([\psi]\).  The
quantity \(|\psi^T\psi|^2\) is independent of the choice of normalized
representative.

\begin{lemma}[Orthogonal orbit normal form]
\label{lem:orbit-normal-form}
Every point of \(\Oqd\) has a normalized representative which, up to an
element of \(SO(d)\), has the form
\begin{equation}
 \psi_q=
 \sqrt{\frac{1+\sqrt q}{2}}\,\ket{1}
 +i\sqrt{\frac{1-\sqrt q}{2}}\,\ket{2}.
 \label{eq:orbit-normal-form}
\end{equation}
Here \(\ket{1},\ket{2}\in\mathcal H\) are a conventional coordinate pair;
any ordered pair of distinct vectors from \(\mathcal B_{\rm O}\) gives an
equivalent normal form.  Consequently \(\Oqd\) is a single projective
\(SO(d)\) orbit, and both the orbit and its normalized invariant measure are
invariant under \(O(d)\).
\end{lemma}

\begin{proof}
Write the coordinate column of \(\psi\) in \(\mathcal B_{\rm O}\) as \(x+iy\),
with \(x,y\in\R^d\).  Multiplication by a global phase acts as a real rotation
of the ordered pair \((x,y)\).  Choose the phase so that \(\psi^T\psi\) is real
and nonnegative.  Then
\begin{equation}
 x\cdot y=0,
 \qquad
 \|x\|^2-\|y\|^2=\sqrt q,
 \qquad
 \|x\|^2+\|y\|^2=1.
\end{equation}
Hence \(\|x\|^2=(1+\sqrt q)/2\) and
\(\|y\|^2=(1-\sqrt q)/2\).  If \(q<1\), the vectors
\(x/\|x\|\) and \(y/\|y\|\) form an ordered orthonormal two-frame.  For
\(d\geq3\), an element of \(SO(d)\) can send this frame to the first two
coordinate axes, represented by \(\ket{1}\) and \(\ket{2}\).  At \(q=1\), one
has \(y=0\), and transitivity of \(SO(d)\) on the unit sphere sends \(x\) to
\(\ket{1}\).  This gives Eq.~\eqref{eq:orbit-normal-form}.  Every \(O\in O(d)\)
preserves \(|\psi^T\psi|\).  Since \(O(d)\) normalizes \(SO(d)\), it also
preserves the unique normalized invariant measure on each orbit.
\end{proof}

The representative in Eq.~\eqref{eq:orbit-normal-form} lies in the
two-dimensional complex subspace
\(\operatorname{span}_{\C}\{\ket{1},\ket{2}\}\), but the action of
\(SO(d)\) moves the corresponding real two-plane through \(\R^d\).
Consequently, the full orbit is in general much larger than this
two-dimensional representative.  A stabilizer count gives \(\operatorname{span}_{\C}\{\ket{1},\ket{2}\}\), while its
real and imaginary parts span at most a real two-plane.  The action of
\(SO(d)\) moves that plane, and hence its complex span, through \(\mathcal H\).
A stabilizer count gives
\begin{equation}
 \dim_{\R}\Oqd=
 \begin{cases}
  2d-4,&q=0,\\
  2d-3,&0<q<1,\\
  d-1,&q=1.
 \end{cases}
 \label{eq:orbit-dimensions}
\end{equation}

We denote the normalized invariant probability measure on \(\Oqd\) by
\(\mu_q\).  At \(q=1\), the vectors are real up to a global phase; at
\(q=0\), their real and imaginary parts are orthogonal and have equal norm.

The vector representation of \(O(d)\) is irreducible over \(\C\), so every
invariant ensemble \((\Oqd,\mu_q)\) is a continuous complex projective
\(1\)-design:
\(\int_{\Oqd}\rho_\psi\,d\mu_q(\psi)=I/d\), where
\(\rho_\psi=\proj{\psi}\).
More generally, a complex projective \(t\)-design is an ensemble whose
average of \(\rho_\psi^{\otimes t}\) equals that for uniformly distributed
complex pure states (the Haar ensemble) \cite{RoyScott2007}; the measure
may be continuous, as it is here.
The orbit \(\mathcal O_0^{(d)}\) is the smooth complex projective quadric
\(Q^{d-2}:=\{[z]\in\mathbb{CP}^{d-1}:z^Tz=0\}\).  Equivalently, it is the
oriented Grassmannian, the space of oriented real two-dimensional planes,
\(\operatorname{Gr}_2^+(\R^d)\simeq
SO(d)/\bigl(SO(2)\times SO(d-2)\bigr)\)
\cite{SuhHwang2016}.  The orbit \(\mathcal O_1^{(d)}\) is
the real projective space \(\mathbb{RP}^{d-1}\), while the orbits with
\(0<q<1\) are principal orbits, the generic orbits of the \(SO(d)\)
action.  This action has cohomogeneity one: its
principal orbits have real codimension one in \(\mathbb{CP}^{d-1}\).

\subsection{Local fidelities, Choi representation, and twirling}

A one-to-two cloner is a completely positive trace-preserving (CPTP) map
\(\Phi:\End(\C^d)\to\End(\C^d\otimes\C^d)\), where
\(\End(\mathcal H)\) denotes the linear operators on \(\mathcal H\).
Both output registers \(A\) and \(B\) have Hilbert space \(\C^d\).
For a fixed channel and input \(\rho_\psi\), write
\begin{equation}
 \begin{aligned}
  \rho_{AB}(\psi)&:=\Phi(\rho_\psi),\\
  \rho_A(\psi)&:=\Tr_B\rho_{AB}(\psi),&
  \rho_B(\psi)&:=\Tr_A\rho_{AB}(\psi).
 \end{aligned}
 \label{eq:output-states}
\end{equation}
For density operators \(\rho\) and \(\sigma\), we use the
Uhlmann--Jozsa fidelity \cite{Uhlmann1976,Jozsa1994}, with the convention
\begin{equation}
 \mathsf F(\rho,\sigma)
 :=\left[\Tr\sqrt{\sqrt{\rho}\,\sigma\sqrt{\rho}}\right]^2.
 \label{eq:state-fidelity}
\end{equation}
Since the target \(\rho_\psi\) is pure,
this reduces to
\(\mathsf F(\rho_\psi,\sigma)=\Tr(\rho_\psi\sigma)
=\langle\psi|\sigma|\psi\rangle\).  The pointwise local, or
single-clone, fidelities are therefore
\cite{KeylWerner1999,CerfAsymmetric2000,IblisdirEtAl2005,
KayRamanathanKaszlikowski2013}
\begin{equation}
 \begin{aligned}
  f_A(\psi;\Phi)&:=\Tr[\rho_\psi\rho_A(\psi)],\\
  f_B(\psi;\Phi)&:=\Tr[\rho_\psi\rho_B(\psi)].
 \end{aligned}
 \label{eq:pointwise-fidelities}
\end{equation}
For an input drawn from \(\mu_q\), the
orbit-averaged local fidelities are
\begin{align}
 F_A^{(d,q)}(\Phi)
 &=\int_{\Oqd}f_A(\psi;\Phi)\,d\mu_q(\psi),\label{eq:FA}\\
 F_B^{(d,q)}(\Phi)
 &=\int_{\Oqd}f_B(\psi;\Phi)\,d\mu_q(\psi).\label{eq:FB}
\end{align}
We suppress the superscripts when \(d\) and \(q\) are fixed.

The achievable fidelity region is
\begin{equation}
 \mathcal F_q^{(d)}=
 \left\{\bigl(F_A^{(d,q)}(\Phi),F_B^{(d,q)}(\Phi)\bigr):
 \Phi\ \text{is CPTP}\right\}\subset[0,1]^2.
 \label{eq:fidelity-region}
\end{equation}
It is compact and convex because the set of finite-dimensional CPTP maps is
compact and convex and the two averaged fidelities are affine in \(\Phi\).
Consequently, the complete region is determined by its support function
\begin{align}
 h_q^{(d)}(a,b)
 &=\max_{(x,y)\in\mathcal F_q^{(d)}}(ax+by)\nonumber\\
 &=\max_{\Phi\ {\rm CPTP}}
 \{aF_A^{(d,q)}(\Phi)+bF_B^{(d,q)}(\Phi)\}.
 \label{eq:support-definition}
\end{align}
Here \(a,b\in\R\).  The region is the intersection of the half-planes
\(ax+by\leq h_q^{(d)}(a,b)\) over all real directions \((a,b)\).
Nonzero nonnegative directions support its cloning-relevant upper boundary;
when \(a,b>0\), every point of the exposed face (the set of maximizers of
\(ax+by\)) is Pareto optimal.  A fidelity pair is Pareto optimal if no
other achievable pair increases one fidelity without decreasing the other.

Figure~\ref{fig:cloning-schematic} depicts the channel and its two reduced
outputs.

\begin{figure}
 \centering
 \begin{tikzpicture}[
   x=1cm,
   y=1cm,
   font=\small,
   wire/.style={->,line width=0.65pt},
   machine/.style={draw,rounded corners=2pt,align=center,
     line width=0.65pt,inner sep=4pt,fill=black!7},
   joint/.style={draw,rounded corners=2pt,align=center,
     text width=3.45cm,line width=0.85pt,inner sep=4pt,
     fill=black!3},
   clone/.style={draw,rounded corners=2pt,align=center,
     text width=2.05cm,line width=0.6pt,inner sep=4pt}
 ]
  \node[align=center] (input) at (0,3.55)
    {unknown input $\rho_\psi$\\[-1pt]
     {\scriptsize $[\psi]\sim\mu_q$ on $\Oqd$}};

  \node[machine,minimum width=1.35cm,minimum height=0.78cm]
    (cloner) at (0,2.25)
    {$\Phi$\\[-1pt]{\scriptsize CPTP $1\!\to\!2$}};
  \draw[wire] (input.south) -- (cloner.north);

  \node[joint,minimum height=0.95cm]
    (jointout) at (0,0.80)
    {\textbf{joint output on $A\otimes B$}\\[-1pt]
     $\rho_{AB}(\psi)=\Phi(\rho_\psi)$\\[-1pt]
     {\scriptsize possibly correlated}};
  \draw[wire] (cloner.south) -- (jointout.north);

  \node[clone,minimum height=1.10cm,
    font=\footnotesize]
    (cloneB) at (-1.32,-1.55)
    {clone $B$\\[-1pt]
     $\rho_B(\psi)$\\[-2pt]
     $=\Tr_A\rho_{AB}(\psi)$\\[-1pt]
     {\scriptsize average fidelity}\\[-2pt]
     {\scriptsize $F_B^{(d,q)}(\Phi)$}};
  \node[clone,minimum height=1.10cm,
    font=\footnotesize]
    (cloneA) at (1.32,-1.55)
    {clone $A$\\[-1pt]
     $\rho_A(\psi)$\\[-2pt]
     $=\Tr_B\rho_{AB}(\psi)$\\[-1pt]
     {\scriptsize average fidelity}\\[-2pt]
     {\scriptsize $F_A^{(d,q)}(\Phi)$}};

  \draw[wire] (jointout.south -| cloneB.north) -- (cloneB.north)
    node[midway,left,font=\scriptsize,align=right]
    {discard $A$};
  \draw[wire] (jointout.south -| cloneA.north) -- (cloneA.north)
    node[midway,right,font=\scriptsize,align=left]
    {discard $B$};
 \end{tikzpicture}
 \caption{Operational setup for asymmetric one-to-two cloning on an
 orthogonal orbit.  The channel maps the unknown input $\rho_\psi$ to
 a joint state on $A\otimes B$, which is not assumed to factorize.  The downward
 branches denote partial traces: discarding $A$ gives clone $B$, while
 discarding $B$ gives clone $A$.  The pointwise overlaps are
 $f_A(\psi;\Phi)$ and $f_B(\psi;\Phi)$; their orbit averages are the
 fidelities in Eqs.~\eqref{eq:FA} and \eqref{eq:FB}.}
 \label{fig:cloning-schematic}
\end{figure}
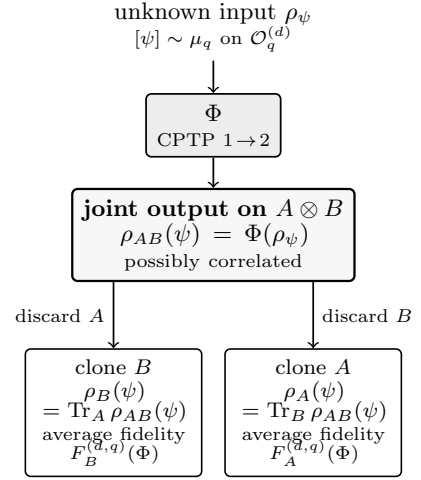

Let \(C'\simeq\C^d\) denote the channel input and let
\(C\simeq\C^d\) be an untouched reference system.  We use the unnormalized
Choi convention \cite{Jamiolkowski1972,Choi1975},
\begin{equation}
 \begin{aligned}
  \ket{\Gamma_d}_{C'C}&=\sum_{j=1}^d\ket{j}_{C'}\ket{j}_C,\\
  J_{\Phi,ABC}
  &=(\Phi_{C'\to AB}\otimes\id_C)(\proj{\Gamma_d}_{C'C}),
 \end{aligned}
 \label{eq:Choi-definition}
\end{equation}
so that
\begin{equation}
 J_\Phi\geq0,
 \qquad
 \Tr_{AB}J_\Phi=I_C,
 \qquad
 \Tr J_\Phi=d.
 \label{eq:choi-conditions}
\end{equation}
Here \(\id_C\) is the identity map on reference operators.  The Choi rank
of \(\Phi\) is \(\rank J_\Phi\), equivalently the minimum number of
operators in a Kraus representation \cite{Choi1975}.
Because the ensemble measure is \(O(d)\)-invariant, an arbitrary channel may be twirled as
\begin{equation}
 \overline\Phi(X)=\int_{O(d)}
 (O\otimes O)\Phi(O^T X O)(O^T\otimes O^T)\,dO.
 \label{eq:Od-twirl}
\end{equation}
Here \(dO\) is the normalized Haar probability measure on \(O(d)\).  The
twirled map is CPTP, has the same two averaged fidelities, and is covariant:
\begin{equation}
 \overline\Phi(OXO^T)
 =(O\otimes O)\overline\Phi(X)(O^T\otimes O^T).
 \label{eq:twirled-covariance}
\end{equation}
For a covariant channel, transitivity also gives
\(f_A(O\psi;\Phi)=f_A(\psi;\Phi)\), and similarly for \(B\).
Each pointwise fidelity is therefore constant on the orbit and equals its
orbit average.  We may optimize over \(O(d)\)-covariant channels without
restricting the achievable fidelity pairs.  The reference
factor of a Choi operator generally carries the conjugate input
representation; here \(\overline O=O\), so
\begin{equation}
 [J_{\overline\Phi},O_A\otimes O_B\otimes O_C]=0
 \qquad (O\in O(d)).
 \label{eq:choi-orthogonal-invariance}
\end{equation}

\subsection{Orthogonal commutant and the Brauer algebra}
\label{sec:brauer-guide}

By Eq.~\eqref{eq:choi-orthogonal-invariance}, the twirled Choi operator on
\(A\otimes B\otimes C\simeq(\C^d)^{\otimes3}\) belongs to
\(\End_{O(d)}((\C^d)^{\otimes3})\), the three-copy orthogonal commutant:
the algebra of all operators commuting with every \(O^{\otimes3}\).
The fixed basis defines the complex-bilinear form
\begin{equation}
 \begin{aligned}
 \mathsf b(u,v)&:=u^Tv=\sum_{j=1}^d u_jv_j,\\
 \mathsf b(Ou,Ov)&=\mathsf b(u,v)
 \qquad\bigl(O\in O(d)\bigr).
 \end{aligned}
 \label{eq:orthogonal-bilinear-form}
\end{equation}
Unlike the Hermitian inner product \(u^\dagger v\), this form contains no
complex conjugation.  In particular, the orbit parameter in
Eq.~\eqref{eq:orbit-parameter} is \(q(\psi)=|\mathsf b(\psi,\psi)|^2\).
The matrix of \(\mathsf b\) in \(\mathcal B_{\rm O}\) is the Kronecker delta,
\(\mathsf b(\ket i,\ket j)=\delta_{ij}\), where \(\delta_{ij}=1\) for
\(i=j\) and \(0\) otherwise.  The corresponding normalized invariant tensor
and contraction operator are
\begin{equation}
 \begin{aligned}
 \ket{\Omega_d}&=\frac1{\sqrt d}\sum_{j=1}^d\ket{j,j},
 &e_d&=d\proj{\Omega_d},\\
 e_d\ket{i,j}&=\delta_{ij}\sum_{k=1}^d\ket{k,k}.&&
 \end{aligned}
 \label{eq:Omega-d}
\end{equation}
Thus \(e_d\) contracts two equal indices and reinserts the invariant tensor.

More generally, the commutant of the diagonal \(O(d)\) action on
\((\C^d)^{\otimes n}\) is generated by permutations of tensor factors and
such contractions.  These operations define the tensor-space representation
of the Brauer algebra \(B_n(d)\), whose image is
\(\End_{O(d)}((\C^d)^{\otimes n})\)
\cite{Brauer1937,Wenzl1988,DotyHu2009}.  This is the orthogonal analogue of
Schur--Weyl duality: for the unitary action on an ordinary tensor power only
permutations occur, whereas orthogonal invariance also permits two tensor
indices to be contracted.

A basis of \(B_n(d)\) is indexed by pairings of \(2n\) vertices arranged in two rows.  Multiplication is defined by concatenating two pairing diagrams and assigning a factor \(d\) to each closed loop.  Figure~\ref{fig:brauer-diagrams} illustrates three-strand diagrams and the loop rule.  In the tensor-space representation, the elementary generators are the swaps \(s_{ij}\) and the contractions
\(e_{ij}=d\proj{\Omega_d}_{ij}\otimes I_{\widehat{ij}}\), where
\(I_{\widehat{ij}}\) acts on the remaining tensor factors.  Thus
\begin{equation}
 s_{ij}^2=I,
 \qquad e_{ij}^2=d e_{ij},
 \qquad s_{ij}e_{ij}=e_{ij}s_{ij}=e_{ij}.
 \label{eq:brauer-basic-relations}
\end{equation}

\begin{figure}
 \centering
 \begin{tikzpicture}[
   x=1.05cm,
   y=0.90cm,
   line width=0.6pt,
   vertex/.style={circle,fill=black,inner sep=1.15pt},
   seam/.style={circle,draw=gray,fill=white,inner sep=0.9pt}
 ]
  \begin{scope}[xshift=0cm]
   \node[vertex] at (0,2.15) {};
   \node[vertex] at (0.42,2.15) {};
   \node[vertex] at (0.84,2.15) {};
   \node[vertex] at (0,1.35) {};
   \node[vertex] at (0.42,1.35) {};
   \node[vertex] at (0.84,1.35) {};
   \draw (0,2.15)--(0,1.35);
   \draw (0.42,2.15)--(0.42,1.35);
   \draw (0.84,2.15)--(0.84,1.35);
   \node at (0.42,1.02) {$I$};
  \end{scope}

  \begin{scope}[xshift=1.65cm]
   \node[vertex] at (0,2.15) {};
   \node[vertex] at (0.42,2.15) {};
   \node[vertex] at (0.84,2.15) {};
   \node[vertex] at (0,1.35) {};
   \node[vertex] at (0.42,1.35) {};
   \node[vertex] at (0.84,1.35) {};
   \draw (0,2.15)--(0.42,1.35);
   \draw (0.42,2.15)--(0,1.35);
   \draw (0.84,2.15)--(0.84,1.35);
   \node at (0.42,1.02) {$s_{12}$};
  \end{scope}

  \begin{scope}[xshift=3.30cm]
   \node[vertex] at (0,2.15) {};
   \node[vertex] at (0.42,2.15) {};
   \node[vertex] at (0.84,2.15) {};
   \node[vertex] at (0,1.35) {};
   \node[vertex] at (0.42,1.35) {};
   \node[vertex] at (0.84,1.35) {};
   \draw (0,2.15) .. controls (0.08,1.86) and (0.34,1.86) .. (0.42,2.15);
   \draw (0,1.35) .. controls (0.08,1.64) and (0.34,1.64) .. (0.42,1.35);
   \draw (0.84,2.15)--(0.84,1.35);
   \node at (0.42,1.02) {$e_{12}$};
  \end{scope}

  \begin{scope}[xshift=0.45cm]
   \node[vertex] at (0,0.60) {};
   \node[vertex] at (0.42,0.60) {};
   \node[vertex] at (0.84,0.60) {};
   \node[vertex] at (0,-0.60) {};
   \node[vertex] at (0.42,-0.60) {};
   \node[vertex] at (0.84,-0.60) {};
   \node[seam] at (0,0) {};
   \node[seam] at (0.42,0) {};
   \node[seam] at (0.84,0) {};
   \draw (0,0.60) .. controls (0.08,0.32) and (0.34,0.32) .. (0.42,0.60);
   \draw (0,0) .. controls (0.08,0.28) and (0.34,0.28) .. (0.42,0);
   \draw (0,0) .. controls (0.08,-0.28) and (0.34,-0.28) .. (0.42,0);
   \draw (0,-0.60) .. controls (0.08,-0.32) and (0.34,-0.32) .. (0.42,-0.60);
   \draw (0.84,0.60)--(0.84,0);
   \draw (0.84,0)--(0.84,-0.60);
   \node at (0.42,-1.13) {$e_{12}\,e_{12}$};
  \end{scope}

  \node at (2.20,0) {$=$};
  \node at (2.70,0) {$d$};

  \begin{scope}[xshift=3.05cm]
   \node[vertex] at (0,0.52) {};
   \node[vertex] at (0.42,0.52) {};
   \node[vertex] at (0.84,0.52) {};
   \node[vertex] at (0,-0.52) {};
   \node[vertex] at (0.42,-0.52) {};
   \node[vertex] at (0.84,-0.52) {};
   \draw (0,0.52) .. controls (0.08,0.24) and (0.34,0.24) .. (0.42,0.52);
   \draw (0,-0.52) .. controls (0.08,-0.24) and (0.34,-0.24) .. (0.42,-0.52);
   \draw (0.84,0.52)--(0.84,-0.52);
   \node at (0.42,-1.13) {$d\,e_{12}$};
  \end{scope}
 \end{tikzpicture}
 \caption{Pairing diagrams in \(B_3(d)\).  Each diagram has three upper
 (output) and three lower (input) vertices, ordered from left to right.
 Upper row: the identity, transposition \(s_{12}\), and contraction
 \(e_{12}\); each leaves the third tensor factor unchanged.  Lower row:
 multiplication by vertical concatenation.  The hollow circles mark the
 identified vertices.  Summing the common index around the closed loop
 contributes \(d\), giving \(e_{12}^2=d e_{12}\).}
 \label{fig:brauer-diagrams}
\end{figure}
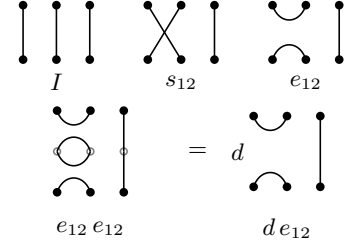

The diagram representation is standard in invariant theory and
orthogonal-group integration \cite{CollinsSniady2006}.  In orthogonally
invariant quantum-spin models, \(s_{ij}\) and \(e_{ij}\) occur as two-site
interaction terms, and Brauer representation theory gives the corresponding
block decomposition used in spectral and thermodynamic analyses
\cite{Ryan2023}.  In quantum information, \(s_{ij}\) exchanges two registers
and \(e_{ij}\) is \(d\) times the projector onto a maximally entangled
vector.  Partial transposition interchanges them,
\begin{equation}
 s_{ij}^{T_j}=e_{ij},
 \label{eq:swap-contraction-partial-transpose}
\end{equation}
Here \(T_j\) denotes partial transposition: it transposes only the \(j\)th
tensor factor in the fixed basis and leaves the others unchanged.
Equation~\eqref{eq:swap-contraction-partial-transpose}
connects the same algebra to marginal compatibility, extendibility, and
entanglement-monogamy problems for orthogonally invariant states
\cite{JohnsonViola2013,SolymosJakabZimboras2026,AllerstorferEtAl2026}.

The cloning objective depends only on the two-copy orbit moment.  Its
partial transpose acts on two factors; embedding it on \(AC\) or \(BC\),
with the identity on the remaining factor, gives an operator in the image
of \(B_3(d)\).

\section{Second moments and complete asymmetric fidelity region}
\label{sec:region}

\subsection{Exact two-copy moment}

Let \(s\) denote the swap on \(\C^d\otimes\C^d\).
Set
\begin{align}
 \Nd&=(d-1)(d+2),
 &\Ad&=d-q,\nonumber\\
 \Cd&=(d+1)q-2.&&
 \label{eq:parameters}
\end{align}

\begin{theorem}[Two-copy moment of an orthogonal orbit]
\label{thm:moment}
For every \(d\geq3\) and \(q\in[0,1]\),
\begin{equation}
 M_2^{(d)}(q)
 :=\int_{\Oqd}\rho_\psi^{\otimes2}\,d\mu_q(\psi)
 =\frac{\Ad(I+s)+\Cd e_d}{d\Nd}.
 \label{eq:general-moment}
\end{equation}
\end{theorem}

\begin{proof}
The symmetric square \(\Sym^2(\C^d)\), consisting of tensors unchanged by
the swap \(s\), decomposes under \(O(d)\) as
\begin{equation}
 \Sym^2(\C^d)=\C\Omega_d\oplus\Sym_0^2(\C^d),
 \qquad
\dim\Sym_0^2(\C^d)=\frac{\Nd}{2}.
\end{equation}
Here \(\Sym_0^2(\C^d)\) denotes symmetric tensors whose coefficient
matrix \(X=(X_{ij})\) is traceless, \(\sum_iX_{ii}=0\).
Equivalently, this is the
orthogonal complement of \(\Omega_d\) in the symmetric square.
The orbit moment is invariant and supported on the symmetric square, so
Schur's lemma implies that it is scalar on each irreducible summand.  Its
weight on the invariant line \(\C\Omega_d\) is
\(\bra{\Omega_d}M_2^{(d)}(q)\ket{\Omega_d}
=\bigl|\langle\Omega_d|\psi^{\otimes2}\rangle\bigr|^2=q/d\).
Therefore
\begin{align}
 M_2^{(d)}(q)
 &=\frac qd\proj{\Omega_d}
 +\frac{1-q/d}{\Nd/2}
 \left(\frac{I+s}{2}-\proj{\Omega_d}\right)\nonumber\\
 &=\frac{(d-q)(I+s)+[(d+1)q-2]e_d}{d(d-1)(d+2)}.
\end{align}
\end{proof}

For the quadric orbit, let
\begin{equation}
 \Pi_0^{(d)}=\frac{I+s}{2}-\proj{\Omega_d}
 \label{eq:Pi0}
\end{equation}
denote the projector onto the traceless symmetric square.  At the three
distinguished values of \(q\), Eq.~\eqref{eq:general-moment} reduces to
\begin{align}
 q=0:\quad&
 M_2^{(d)}(0)=\frac{2}{\Nd}\Pi_0^{(d)},
 \label{eq:null-moment}\\
 q=q_{\rm H}:=\frac{2}{d+1}:\quad&
 M_2^{(d)}(q_{\rm H})=\frac{I+s}{d(d+1)},
 \label{eq:Haar-moment}\\
 q=1:\quad&
 M_2^{(d)}(1)=\frac{I+s+e_d}{d(d+2)}.
 \label{eq:real-moment}
\end{align}
At \(q=q_{\rm H}\), this is the complex Haar second moment.  Hence
\(\mathcal O_{q_{\rm H}}^{(d)}\) is a continuous complex projective
\(2\)-design \cite{RoyScott2007}, and its averaged local-fidelity cloning
problem coincides with universal \(d\)-dimensional cloning.

The quadric-orbit moment is maximally mixed on a codimension-one subspace of
the symmetric square.  Since \(\rank\Pi_0^{(d)}=\Nd/2\),
\begin{equation}
 \frac12\left\|
 M_2^{(d)}(0)-\frac{I+s}{d(d+1)}
 \right\|_1=\frac{2}{d(d+1)},
 \label{eq:null-Haar-distance}
\end{equation}
where \(\|X\|_1:=\Tr\sqrt{X^\dagger X}\) is the trace norm.
The second frame potential is the average fourth power of the overlap
between two independently sampled states:
\begin{equation}
 \iint |\langle\psi|\phi\rangle|^4\,
 d\mu_q(\psi)d\mu_q(\phi)
 =\Tr\!\left[M_2^{(d)}(q)^2\right].
 \label{eq:frame-potential}
\end{equation}
At \(q=0\), it equals \(2/\Nd\), compared with the Haar value
\(2/[d(d+1)]\) \cite{RoyScott2007}.

By Eq.~\eqref{eq:swap-contraction-partial-transpose}, partial transposition
exchanges \(s\) and \(e_d\).  Hence the one-clone fidelity operator is
\begin{equation}
 R_q^{(d)}=
\int_{\Oqd}\rho_\psi\otimes\rho_\psi^T\,d\mu_q(\psi)
 =\frac{\Ad(I+e_d)+\Cd s}{d\Nd}.
 \label{eq:general-R}
\end{equation}
Here and below, a two-body operator carrying subscripts \(AC\) or \(BC\) is
understood as an operator on \(A\otimes B\otimes C\), with the identity on
the unused output factor.  In our unnormalized convention, the Choi
reconstruction formula is
\begin{equation}
 \Phi(X)=\Tr_C\!\left[J_\Phi\bigl(I_{AB}\otimes X_C^T\bigr)\right].
 \label{eq:Choi-reconstruction}
\end{equation}
Equations~\eqref{eq:FA}--\eqref{eq:FB} and
\eqref{eq:Choi-reconstruction} yield
\begin{equation}
 \begin{aligned}
 F_A^{(d,q)}(\Phi)&=\Tr\!\left[J_\Phi R_{q,AC}^{(d)}\right],\\
 F_B^{(d,q)}(\Phi)&=\Tr\!\left[J_\Phi R_{q,BC}^{(d)}\right].
 \end{aligned}
 \label{eq:fidelities-Choi-R}
\end{equation}

\subsection{Three-copy Brauer reduction}

For real weights \(a,b\), define
\begin{equation}
 K_q^{(d)}(a,b)=aR_{q,AC}^{(d)}+bR_{q,BC}^{(d)}.
 \label{eq:general-K}
\end{equation}
Equation~\eqref{eq:fidelities-Choi-R} implies
\(aF_A^{(d,q)}(\Phi)+bF_B^{(d,q)}(\Phi)
=\Tr[J_\Phi K_q^{(d)}(a,b)]\).
The resulting optimization over positive Choi matrices with a fixed
partial trace is a semidefinite program.  The following lemma shows that
its optimum is given by the largest eigenvalue of \(K_q^{(d)}(a,b)\):
\begin{equation}
 h_q^{(d)}(a,b)=d\lambda_{\max}\!\left[K_q^{(d)}(a,b)\right].
 \label{eq:eigenvalue-reduction}
\end{equation}

\begin{lemma}[Spectral achievability]
\label{lem:spectral-achievability}
Let \(P_{\max}\) be the spectral projector onto the maximal eigenspace of
\(K_q^{(d)}(a,b)\), and let \(r_{\max}=\rank P_{\max}\).  Then
\begin{equation}
 J_{a,b}=\frac{d}{r_{\max}}P_{\max}
 \label{eq:top-projector}
\end{equation}
is feasible and attains the optimum in
Eq.~\eqref{eq:eigenvalue-reduction}.
\end{lemma}

\begin{proof}
The projector \(P_{\max}\) commutes with \(O^{\otimes3}\).  Taking the
partial trace and using its invariance under conjugation on \(A\) and
\(B\) shows that \(\Tr_{AB}P_{\max}\) commutes with \(O_C\).  The vector
representation on \(C\) is irreducible, so this partial trace equals
\(\alpha I_C\).  Taking traces gives \(d\alpha=r_{\max}\), so
Eq.~\eqref{eq:top-projector} is positive and satisfies the trace-preservation
constraint.  Its objective value is \(d\lambda_{\max}\).  Conversely,
every feasible \(J\) has trace \(d\), and
\begin{align}
 \Tr[J K_q^{(d)}(a,b)]
 &\leq\lambda_{\max}[K_q^{(d)}(a,b)]\Tr J\nonumber\\
 &=d\lambda_{\max}[K_q^{(d)}(a,b)].
\end{align}
\end{proof}

Identify tensor factors \(1,2,3\) with \(A,B,C\), respectively.  Define
embeddings \(T_\alpha:\C^d\to(\C^d)^{\otimes3}\) by
\begin{align}
 T_1(v)&=\frac1{\sqrt d}\sum_{j=1}^d
 \ket v_A\otimes\ket j_B\otimes\ket j_C,\nonumber\\
 T_2(v)&=\frac1{\sqrt d}\sum_{j=1}^d
 \ket j_A\otimes\ket v_B\otimes\ket j_C,\nonumber\\
 T_3(v)&=\frac1{\sqrt d}\sum_{j=1}^d
 \ket j_A\otimes\ket j_B\otimes\ket v_C.
 \label{eq:Tmaps}
\end{align}
For \(X=\sum_{i,j,k}X_{ijk}\ket{i,j,k}\), define contraction over factors
\(2,3\) by
\(\mathsf c_{23}(X)=\sum_{i,j}X_{ijj}\ket i\), and define
\(\mathsf c_{13}\) and \(\mathsf c_{12}\) analogously.  These maps contract
tensor indices using the invariant form \(\delta_{ij}\); they are not
density-operator partial traces.  The embeddings in Eq.~\eqref{eq:Tmaps}
are their normalized adjoints; for example,
\(T_1=d^{-1/2}\mathsf c_{23}^\dagger\).  Define
\begin{equation}
 \begin{aligned}
 \mathcal T_d
 &:=\sum_{\alpha=1}^3\operatorname{Ran}T_\alpha\\
 &=\sum_{1\leq j<k\leq3}\operatorname{Ran}\mathsf c_{jk}^\dagger.
 \end{aligned}
 \label{eq:trace-sector-definition}
\end{equation}
Each \(T_\alpha\) intertwines the vector and three-copy actions, meaning
\(O^{\otimes3}T_\alpha=T_\alpha O\).
The space \(\mathcal T_d\) is the vector isotypic component of
\((\C^d)^{\otimes3}\) under
\(O(d)\), namely the sum of all submodules isomorphic to the defining vector
representation.  The maps \(T_1,T_2,T_3\) give a nonorthogonal basis of its
three-dimensional multiplicity space: its coordinates label the three
equivalent copies of the vector representation.  Equivalently, \(\mathcal T_d\) is
obtained by inserting the invariant metric into any pair of indices of a
rank-three tensor.  Its orthogonal complement is the traceless, or harmonic,
tensor subspace \cite{DotyHu2009},
\begin{equation}
 \mathcal T_d^\perp
 =\ker\mathsf c_{23}\cap\ker\mathsf c_{13}\cap\ker\mathsf c_{12}.
 \label{eq:traceless-sector}
\end{equation}
Since \(e_{jk}=\mathsf c_{jk}^\dagger\mathsf c_{jk}\), every contraction
operator vanishes on \(\mathcal T_d^\perp\).  The Gram matrix of the three
embeddings in multiplicity space is
\begin{equation}
 \begin{aligned}
 T_\alpha^\dagger T_\beta&=(G_d)_{\alpha\beta}I_{\C^d},\\
 G_d&=
 \begin{pmatrix}
 1&1/d&1/d\\
 1/d&1&1/d\\
 1/d&1/d&1
 \end{pmatrix}>0.
 \end{aligned}
 \label{eq:general-G}
\end{equation}
Its eigenvalues are \(1+2/d\) and \(1-1/d\), the latter with multiplicity
two.  Thus it is positive for \(d\geq3\), and the ranges of the three
embeddings span the \(3d\)-dimensional vector isotypic component
\(\mathcal T_d\).  In the
ordered multiplicity coordinates \((T_1,T_2,T_3)\), the matrix representing
\(d\Nd K_q^{(d)}(a,b)\) is
\begin{equation}
 \resizebox{0.98\columnwidth}{!}{$\displaystyle
 B_q^{(d)}(a,b)=
 \begin{pmatrix}
 \Ad a+\Nd b&\Ad b&\Cd a+\Ad b\\
 \Ad a&\Nd a+\Ad b&\Ad a+\Cd b\\
 \Cd a&\Cd b&\Ad(a+b)
 \end{pmatrix}. $}
 \label{eq:general-B}
\end{equation}
Since the multiplicity basis is not orthonormal,
\(B_q^{(d)}(a,b)\) need not be a symmetric matrix.  Instead,
\(\bigl(B_q^{(d)}(a,b)\bigr)^TG_d=G_dB_q^{(d)}(a,b)\).
The real matrix \(B_q^{(d)}(a,b)\) is therefore self-adjoint for the
Hermitian inner product
\(\langle x,y\rangle_{G_d}=x^\dagger G_dy\) on \(\C^3\).  Equivalently,
\(G_d^{1/2}B_q^{(d)}(a,b)G_d^{-1/2}\) is real symmetric.  Hence its
eigenvalues are real, and its eigenvectors may be chosen with real
multiplicity coordinates.  Appendix~\ref{app:spectrum} derives this block
and the spectrum on its orthogonal complement.

\subsection{Complete support function}

On \(\mathcal T_d^\perp\), only permutations remain.  The symmetric group
\(S_3\), which permutes the three tensor factors, has the trivial and sign
representations on fully symmetric and fully antisymmetric tensors, and a
two-dimensional standard representation on tensors of mixed permutation
symmetry.  These yield four eigenvalue branches in addition to the three
on the vector multiplicity space:

\begin{theorem}[Complete orthogonal-orbit support function]
\label{thm:support}
Let \(s_1=a+b\) and \(r=\sqrt{a^2+b^2-ab}\), and let
\(\mu_1,\mu_2,\mu_3\) be the eigenvalues of \(B_q^{(d)}(a,b)\).  Then, for
every \(d\geq3\), \(q\in[0,1]\), and \(a,b\in\R\),
\begin{equation}
 h_q^{(d)}(a,b)=\frac1{\Nd}
 \max\left\{
 \begin{array}{l}
 \mu_1,\mu_2,\mu_3,\\
 (d-2+dq)s_1,\\
 (d+2)(1-q)s_1,\\
 \Ad s_1+\Cd r,\\
 \Ad s_1-\Cd r
 \end{array}
 \right\}.
 \label{eq:general-support}
\end{equation}
The characteristic polynomial of \(B_q^{(d)}(a,b)\) is given in
Eq.~\eqref{eq:general-cubic} of Appendix~\ref{app:spectrum}.
\end{theorem}

\begin{proof}
By Eq.~\eqref{eq:traceless-sector}, the contraction operators vanish on
\(\mathcal T_d^\perp\), and
\(d\Nd K_q^{(d)}(a,b)=\Ad(a+b)I+\Cd(as_{13}+bs_{23})\).
The trivial and sign representations of \(S_3\) give \((\Ad+\Cd)s_1=(d-2+dq)s_1\) and \((\Ad-\Cd)s_1=(d+2)(1-q)s_1\).  On the standard representation, \(as_{13}+bs_{23}\) has eigenvalues \(\pm r\), producing the last two branches.  The block on the vector multiplicity space is Eq.~\eqref{eq:general-B}; expansion of \(\det(\mu I-B_q^{(d)})\) gives Eq.~\eqref{eq:general-cubic}.  Lemma~\ref{lem:spectral-achievability} converts the largest spectral value into an attainable Choi optimum.
\end{proof}

\begin{corollary}[Dominance of the vector isotypic component for nonnegative weights]
\label{cor:positive}
For \(a,b\geq0\), not both zero, the maximizing branch in Eq.~\eqref{eq:general-support} lies in the vector isotypic component.  If \(\mu_+(d,q;a,b)\) is the largest root of Eq.~\eqref{eq:general-cubic}, then
\begin{equation}
 h_q^{(d)}(a,b)=\frac{\mu_+(d,q;a,b)}{\Nd}.
 \label{eq:positive-support}
\end{equation}
\end{corollary}

\begin{proof}
Put \(M=\max\{a,b\}\).  Since \(R_q^{(d)}\geq0\) and
\(R_q^{(d)}\ket{\Omega_d}=\ket{\Omega_d}/d\),
\(\lambda_{\max}[K_q^{(d)}(a,b)]\geq M/d\).  Every complementary
eigenvalue has magnitude at most
\((\Ad+|\Cd|)(a+b)/(d\Nd)\leq2M(\Ad+|\Cd|)/(d\Nd)\).
For \(d\geq4\), \(2(\Ad+|\Cd|)\leq4d-4<\Nd\), so this is strictly smaller than \(M/d\).  For \(d=3\), equality in the bound can occur only for \(q=0\) and \(a=b=M\).  For the rescaled operator \(d\Nd K_q^{(d)}\), the largest complementary branch is then \(10M\), whereas the largest eigenvalue on the vector isotypic component is \((11+\sqrt{21})M>10M\).  Thus the vector isotypic component dominates for every \(d\geq3\).
\end{proof}

\section{Optimal channels and special orthogonal orbits}
\label{sec:channels}

\subsection{Explicit covariant channels in arbitrary dimension}

For \(c=(c_1,c_2,c_3)^T\in\R^3\) satisfying
\begin{equation}
 c^T G_d c=1,
 \label{eq:cnorm}
\end{equation}
define \(d\) Kraus operators \(K_k:\C^d\to\C^d\otimes\C^d\) by
\begin{align}
 K_k={}&\frac{c_1}{\sqrt d}\sum_{j=1}^d\ket{k,j}\bra j
 +\frac{c_2}{\sqrt d}\sum_{j=1}^d\ket{j,k}\bra j\nonumber\\
 &+\frac{c_3}{\sqrt d}\sum_{j=1}^d\ket{j,j}\bra k,
 \qquad k=1,\ldots,d.
 \label{eq:general-kraus}
\end{align}

Define
\begin{align}
 M_A^{(d)}(q)&=
 \begin{pmatrix}
 1/d&1/d&q/d\\
 1/d&1&1/d\\
 q/d&1/d&1/d
 \end{pmatrix},\label{eq:MA}\\[1mm]
 M_B^{(d)}(q)&=
 \begin{pmatrix}
 1&1/d&1/d\\
 1/d&1/d&q/d\\
 1/d&q/d&1/d
 \end{pmatrix}.\label{eq:MB}
\end{align}

\begin{theorem}[Explicit optimal covariant channels]
\label{thm:kraus}
The map
\begin{equation}
 \Phi_c(\rho)=\sum_{k=1}^dK_k\rho K_k^\dagger
 \label{eq:Phic}
\end{equation}
is an \(O(d)\)-covariant CPTP channel.  Its Choi matrix \(J_{\Phi_c}\) is an
orthogonal projector of rank \(d\).  Write
\(F_A^{(d,q)}(c):=F_A^{(d,q)}(\Phi_c)\), and similarly for \(B\).
Its local fidelities on \(\Oqd\) are
\begin{equation}
 F_A^{(d,q)}(c)=c^TM_A^{(d)}(q)c,
 \qquad
 F_B^{(d,q)}(c)=c^TM_B^{(d)}(q)c.
\label{eq:FM}
\end{equation}
For every nonnegative supporting direction \((a,b)\neq(0,0)\), an optimal
channel is obtained by choosing \(c\) as a normalized top generalized
eigenvector of
\begin{equation}
 [aM_A^{(d)}(q)+bM_B^{(d)}(q)]c
 =h_q^{(d)}(a,b)G_dc.
 \label{eq:generalized-eigenproblem}
\end{equation}
For \(a,b>0\), its fidelity pair is Pareto optimal.
\end{theorem}

\begin{proof}
Equation~\eqref{eq:Phic} is a Kraus representation, so \(\Phi_c\) is
completely positive.  Appendix~\ref{app:spectrum},
Eq.~\eqref{eq:kraus-completeness-calculation}, gives
\(\sum_{k=1}^dK_k^\dagger K_k=(c^TG_dc)I=I\), where
Eq.~\eqref{eq:cnorm} gives the last equality.  Hence \(\Phi_c\) is
trace preserving.  Define the map
\begin{equation}
 \begin{aligned}
 W_c&:=c_1T_1+c_2T_2+c_3T_3,\\
 W_c v&=c_1\ket v_A\ket{\Omega_d}_{BC}
 +c_2\ket{\Omega_d}_{AC}\ket v_B
 +c_3\ket{\Omega_d}_{AB}\ket v_C,
 \end{aligned}
 \label{eq:Wc}
\end{equation}
where tensor-factor labels indicate the location of \(v\).  The Kraus
operators in Eq.~\eqref{eq:general-kraus} satisfy
\(W_c\ket{k}_C=(K_k\otimes I_C)\ket{\Gamma_d}_{C'C}\),
\(k=1,\ldots,d\), and therefore
\begin{equation}
 J_{\Phi_c}
 =\sum_{k=1}^d(K_k\otimes I_C)\proj{\Gamma_d}_{C'C}
 (K_k^\dagger\otimes I_C)
 =W_cW_c^\dagger.
 \label{eq:Wc-Choi-relation}
\end{equation}
Each embedding \(T_\alpha\) is an \(O(d)\)-intertwiner, so
\(O^{\otimes3}W_c=W_cO\).  Equation~\eqref{eq:Wc-Choi-relation} then shows
that the Choi matrix is invariant and the channel is covariant.  Moreover,
\(W_c^\dagger W_c=(c^TG_dc)I=I\), so \(W_c\) is an isometry and
\(\rank J_{\Phi_c}=d\).  Direct evaluation against
\(R_{q,AC}^{(d)}\) and \(R_{q,BC}^{(d)}\) gives Eq.~\eqref{eq:FM}, with
the matrices in Eqs.~\eqref{eq:MA} and \eqref{eq:MB}.  Finally,
\begin{equation}
 aM_A^{(d)}(q)+bM_B^{(d)}(q)
 =\frac1{\Nd}G_dB_q^{(d)}(a,b),
 \label{eq:pencil-identity}
\end{equation}
so the generalized eigenproblem coincides with the spectral problem on the
vector isotypic component in Theorem~\ref{thm:support}.
\end{proof}

When the top generalized eigenvalue is simple, \(c\) is unique up to the
sign \(c\mapsto-c\), which leaves \(\Phi_c\) unchanged.  If the maximal
eigenspace is degenerate, convex combinations of the corresponding optimal
channels remain optimal and may have Choi rank larger than \(d\).

Define the two reduced channels by \(\Phi_{c,A}:=\Tr_B\circ\Phi_c\) and
\(\Phi_{c,B}:=\Tr_A\circ\Phi_c\).
They have the orthogonally covariant form
\begin{align}
 \Phi_{c,A}(\rho)
 &=\alpha_A\rho+\beta_A\rho^T+\gamma_A\Tr(\rho)I,\label{eq:marginal-A}\\
 \Phi_{c,B}(\rho)
 &=\alpha_B\rho+\beta_B\rho^T+\gamma_B\Tr(\rho)I,\label{eq:marginal-B}
\end{align}
with
\begin{align}
 \alpha_A&=c_2^2+\frac{2}{d}(c_1c_2+c_2c_3),
 &\beta_A&=\frac{2c_1c_3}{d},\nonumber\\
 \gamma_A&=\frac{c_1^2+c_3^2}{d},&&\nonumber\\
 \alpha_B&=c_1^2+\frac{2}{d}(c_1c_2+c_1c_3),
 &\beta_B&=\frac{2c_2c_3}{d},\nonumber\\
 \gamma_B&=\frac{c_2^2+c_3^2}{d}.&&
 \label{eq:marginal-coefficients}
\end{align}
These coefficients obey
\(\alpha_A+\beta_A+d\gamma_A=\alpha_B+\beta_B+d\gamma_B=c^TG_dc=1\), as
required for trace preservation.  They also give
\(\alpha_A+q\beta_A+\gamma_A=c^TM_A^{(d)}(q)c\) and
\(\alpha_B+q\beta_B+\gamma_B=c^TM_B^{(d)}(q)c\).  For
\([\psi]\in\mathcal O_q^{(d)}\),
\(\Tr(\rho_\psi\rho_\psi^T)=|\psi^T\psi|^2=q\); hence the dependence on
\(q\) enters through the transpose terms in the reduced channels.

\subsection{Symmetric, universal, real-state, and minimal-fidelity points}

A channel is exchange symmetric if
\(s_{AB}\Phi(X)s_{AB}=\Phi(X)\) for every input operator \(X\).
This implies equal marginal channels and equal fidelities; equal average
fidelities alone do not imply exchange symmetry.
At \(a=b=1\), exchanging the two outputs sends
\((c_1,c_2,c_3)\) to \((c_2,c_1,c_3)\).  The antisymmetric vector
\((1,-1,0)\) has eigenvalue \(\Nd\).  With respect to the ordered basis
\(\bigl((1,1,0),(0,0,1)\bigr)\)
of the exchange-symmetric multiplicity subspace, which is nonorthonormal
in the \(G_d\) inner product, the restriction of \(B_q^{(d)}(1,1)\)
is represented by
\begin{equation}
 \begin{pmatrix}
 \Nd+2\Ad&\Ad+\Cd\\
 2\Cd&2\Ad
 \end{pmatrix}.
 \label{eq:symmetric-multiplicity-block}
\end{equation}
Writing
\(\Delta=\Nd^2+8\Cd(\Ad+\Cd)\), its larger eigenvalue is
\(\mu_+=(\Nd+4\Ad+\sqrt\Delta)/2\).  It is strictly larger than the
antisymmetric eigenvalue because
\begin{equation}
 \Delta-(\Nd-4\Ad)^2=8\Nd(d-2+q^2)>0.
 \label{eq:symmetric-sector-dominance}
\end{equation}
Thus \(\mu_+\) is the top eigenvalue on the vector isotypic component, and
Corollary~\ref{cor:positive} gives
\(h_q^{(d)}(1,1)=\mu_+/\Nd\).  Moreover, output symmetrization replaces any
channel \(\Phi\) by
\(\Phi^{\rm sym}(X):=[\Phi(X)+s_{AB}\Phi(X)s_{AB}]/2\) and gives
\(F_A(\Phi^{\rm sym})=F_B(\Phi^{\rm sym})
 =[F_A(\Phi)+F_B(\Phi)]/2\).  The optimal common local fidelity is therefore
\begin{equation}
 F_{\rm sym}^{(d)}(q):=\frac12h_q^{(d)}(1,1).
 \label{eq:symmetric-fidelity-definition}
\end{equation}

\begin{corollary}[Symmetric fidelity]
\label{cor:symmetric-fidelity}
For every \(d\geq3\),
\begin{equation}
 F_{\rm sym}^{(d)}(q)=
 \frac{\Nd+4\Ad+
 \sqrt{\Nd^2+8\Cd(\Ad+\Cd)}}{4\Nd}.
 \label{eq:symmetric-fidelity}
\end{equation}
The unique minimum occurs at \(q=q_{\rm H}=2/(d+1)\) and equals the universal value
\begin{equation}
 F_{\rm sym}^{(d)}(q_{\rm H})=\frac{d+3}{2(d+1)}.
 \label{eq:universal-symmetric}
\end{equation}
\end{corollary}

\begin{proof}
Put \(\Delta=\Nd^2+8\Cd(\Ad+\Cd)\) and
\(L=(d+1)(\Ad+\Cd)+d\Cd\).  Then
\begin{equation}
 4\Nd\frac{d}{dq}F_{\rm sym}^{(d)}(q)
 =-4+\frac{4L}{\sqrt\Delta},
\end{equation}
and direct algebra gives \(L^2-\Delta=4\Nd(\Ad+\Cd)\Cd\).
Here \(\Ad+\Cd=d-2+dq>0\).  If \(\Cd>0\), then also \(L>0\), so \(L>\sqrt\Delta\).  If \(\Cd<0\), either \(L\leq0\) or \(0<L<\sqrt\Delta\).  Thus the derivative has the sign of \(\Cd\), equivalently the sign of \(q-2/(d+1)\), and it vanishes only at the stated point.
\end{proof}

The endpoint values are
\begin{align}
 F_{\rm sym}^{(d)}(0)
 &=\frac{\Nd+4d+\sqrt{\Nd^2-16(d-2)}}{4\Nd},
 \label{eq:null-symmetric-fidelity}\\
 F_{\rm sym}^{(d)}(1)
 &=\frac{d+6+\sqrt{d^2+4d+20}}{4(d+2)}.
 \label{eq:real-symmetric-fidelity}
\end{align}
The value in Eq.~\eqref{eq:real-symmetric-fidelity} agrees with the known
symmetric real-state cloning fidelity
\cite{NavezCerf2003,ZhangYe2007}.

At the projective \(2\)-design orbit, write
\(h_{\rm H}^{(d)}:=h_{q_{\rm H}}^{(d)}\).  For nonnegative weights it
simplifies to
\begin{multline}
 h_{\rm H}^{(d)}(a,b)=
 \frac{(d+2)(a+b)}{2(d+1)}\\
 {}+\frac{\sqrt{d^2(a-b)^2+4ab}}{2(d+1)},
 \qquad a,b\geq0.
 \label{eq:universal-support}
\end{multline}
This is the standard universal \(1\to2\) asymmetric support function \cite{KayRamanathanKaszlikowski2013,Hashagen2017,NechitaPellegriniRochette2021,NechitaPellegriniRochette2023}.

The smallest attainable fidelity of either marginal is
\begin{equation}
 F_{\min}^{(d)}(q)
 :=\min_{\Phi\ \mathrm{CPTP}}F_A^{(d,q)}(\Phi)
 =\min_{\Phi\ \mathrm{CPTP}}F_B^{(d,q)}(\Phi),
 \label{eq:marginal-minimum-definition}
\end{equation}
where both minima are over one-to-two channels.  Exchanging the outputs
shows that they are equal.  For every \(d\geq3\),
\begin{equation}
 \begin{aligned}
 F_{\min}^{(d)}(q)&=\frac{\Ad-|\Cd|}{\Nd}\\
 &=\begin{cases}
 \dfrac{d-2+dq}{\Nd},
 &0\leq q\leq\dfrac{2}{d+1},\\[2mm]
 \dfrac{(d+2)(1-q)}{\Nd},
 &\dfrac{2}{d+1}\leq q\leq1.
 \end{cases}
 \end{aligned}
 \label{eq:marginal-minimum}
\end{equation}
Both outputs can attain this minimum simultaneously:
\begin{equation}
 (F_A,F_B)_{\rm min}^{(d,q)}
 =\bigl(F_{\min}^{(d)}(q),F_{\min}^{(d)}(q)\bigr).
 \label{eq:componentwise-minimum}
\end{equation}
The perfect-copy endpoints are \((1,1/d)\) and \((1/d,1)\).
Appendix~\ref{app:marginal-extrema} proves these statements.

\subsection{Brauer-state compatibility}
\label{sec:brauer}

For an \(O(d)\)-covariant channel, the normalized Choi state
\(\omega_{ABC}=J_\Phi/d\) satisfies \(\omega_C=I_C/d\).
Its marginals \(\sigma_{AC}:=\omega_{AC}\) and
\(\tau_{BC}:=\omega_{BC}\) lie in
\(\operatorname{span}\{I,s,e_d\}\).  Such orthogonally invariant
bipartite states are called Brauer states
\cite{SolymosJakabZimboras2026}.  Thus the asymmetric cloning problem can also be viewed as a marginal
compatibility problem for two, in general different, orthogonally invariant
bipartite states (Brauer states)
\cite{SolymosJakabZimboras2026,AllerstorferEtAl2026,JohnsonViola2013}.

For an exchange-symmetric cloner, \(\sigma_{AC}\) and \(\tau_{BC}\)
are equal after identifying \(A\simeq B\).  The common marginal is
\((1,2)\)-extendible: it admits an extension with two identical output
marginals that is invariant under \(A\leftrightarrow B\).
The indices refer to the bipartition \(C{:}A\).
Conversely, every \((1,2)\)-extendible Brauer state occurs as the common
marginal of such a cloner, as proved in
Appendix~\ref{app:brauer-extensions}.
The symmetric fidelity in Eq.~\eqref{eq:symmetric-fidelity} is therefore
the maximum of \(d\Tr[\sigma_{AC}R_q^{(d)}]\) over the extendible
Brauer states characterized in Ref.~\cite{SolymosJakabZimboras2026}.

\section{Four-mode fermionic realization}
\label{sec:fermions}

\subsection{Majoranas, triality, and concurrence}

Let \(a_j,a_j^\dagger\), \(j=1,\ldots,4\), be fermionic annihilation and
creation operators, and define the eight Hermitian Majorana operators by
\(\gamma_{2j-1}:=a_j+a_j^\dagger\) and
\(\gamma_{2j}:=i(a_j-a_j^\dagger)\).
They satisfy
\begin{equation}
 \{\gamma_r,\gamma_s\}=2\delta_{rs}I.
 \label{eq:CAR}
\end{equation}
The number operator \(\widehat N:=\sum_{j=1}^4a_j^\dagger a_j\) defines the
fermion-parity operator \(P_4:=(-1)^{\widehat N}\), which splits the
sixteen-dimensional Fock space into two eight-dimensional sectors.  Write
\(\ket{\mathbf n}_{\rm F}:=\ket{n_1n_2n_3n_4}_{\rm F}\), where
\(\mathbf n\in\{0,1\}^4\).  We work in the even sector
\(\mathcal H_+\simeq\C^8\), whose occupation-number basis is
\begin{equation}
 \mathcal B_{\rm F}=
 \left\{\ket{\mathbf n}_{\rm F}:\sum_{j=1}^4 n_j\equiv0\pmod 2\right\}.
 \label{eq:even-fock-basis}
\end{equation}

Up to an overall phase, a quadratic fermionic unitary is generated by a
Hamiltonian that is a real linear combination of \(i\gamma_r\gamma_s\),
\(r<s\).  Such unitaries preserve parity.  A pure state is fermionic Gaussian
if it can be obtained from a Fock-basis state by such a unitary.
For a normalized pure state
\(\ket\Psi\in\mathcal H_+\), define
\(\Gamma_{rs}(\Psi):=(i/2)\bra\Psi[\gamma_r,\gamma_s]\ket\Psi\).  The covariance
matrix \(\Gamma(\Psi)\) is real and antisymmetric, and \(\ket\Psi\) is
Gaussian if and only if \(\Gamma(\Psi)^2=-I\).
Thus its covariance matrix is a real orthogonal complex structure: a real
orthogonal matrix whose square is \(-I\)
\cite{Bravyi2005,KrausWolf2009,deMeloCwikTerhal2013}.
A mixed Gaussian state is one that a quadratic unitary can transform into
a product of single-mode occupation-number mixtures \cite{Bravyi2005}.
Fermionic linear-optics (FLO) unitaries act on the full Fock space.  We use
the Heisenberg convention for their action on Majoranas,
\begin{equation}
 U^\dagger\gamma_rU=\sum_{s=1}^8R_{rs}\gamma_s,
 \qquad R\in SO(8).
 \label{eq:FLO}
\end{equation}
Up to an overall phase, these unitaries form a representation \(U(g)\) of
\(\mathrm{Spin}(8)\), the double cover of \(SO(8)\).  Its actions on the
even- and odd-parity sectors are the two half-spin representations; write
\(U_+(g)\) for the even-sector restriction.  Triality is an outer
automorphism of \(\mathrm{Spin}(8)\) that permutes its vector and two
half-spin representations; ``outer'' means that it is not conjugation by
an element of the group.  Choose such an automorphism \(\tau\) and a
unitary intertwiner
\cite{Chevalley1954,KnusParimalaSridharan1994}
\begin{equation}
 \begin{aligned}
 T_{\rm tr}&:\mathcal H_+\longrightarrow\C^8,\\
 T_{\rm tr}U_+(g)T_{\rm tr}^\dagger
 &=\rho_{\rm v}(\tau(g))\in SO(8),
 \qquad g\in\mathrm{Spin}(8),
 \end{aligned}
 \label{eq:triality-intertwiner}
\end{equation}
where \(\rho_{\rm v}\) is the vector representation.  For \(U=U(g)\),
the Majorana transformation in Eq.~\eqref{eq:FLO} is
\(R=\rho_{\rm v}(g)\), whereas the matrix in
Eq.~\eqref{eq:triality-intertwiner} is \(\rho_{\rm v}(\tau(g))\).
These generally differ because of the triality automorphism.

Let \(K_{\rm O}\) denote complex conjugation in the real orthonormal basis
\(\mathcal B_{\rm O}\) of Sec.~\ref{sec:orthogonal-problem}.  The
antiunitary \(\theta_+:=T_{\rm tr}^\dagger K_{\rm O}T_{\rm tr}\) is the
\(\mathrm{Spin}(8)\)-invariant conjugation, or real structure, on the even
sector: \(\theta_+^2=I\) and
\(\theta_+U_+(g)=U_+(g)\theta_+\).  It defines the generalized fermionic
concurrence used below \cite{OszmaniecGuttKus2014}.  The vectors
\(T_{\rm tr}^\dagger\ket{j}_{\rm O}\), \(j=1,\ldots,8\), form an orthonormal
basis of \(\mathcal H_+\) fixed by \(\theta_+\).  This basis, induced by the
chosen intertwiner, is generally different from the occupation-number basis
\(\mathcal B_{\rm F}\).  The map \(T_{\rm tr}\) is an identification of
representation spaces, not a physical FLO transformation.

Lemma~\ref{lem:orbit-normal-form} implies that each projective \(SO(8)\)
orbit and its normalized invariant measure are invariant under \(O(8)\).
The twirl may therefore be taken over \(O(8)\), although
orientation-reversing elements are not physical FLO transformations.

For a physical state \(\ket{\Psi}\in\mathcal H_+\), let
\(\psi_{\rm tr}\in\C^8\) be the coordinate column of
\(T_{\rm tr}\ket{\Psi}\).  Its fermionic concurrence and the orthogonal-orbit
parameter are
\begin{equation}
 \begin{aligned}
 C_+(\Psi)&:=|\bra\Psi\theta_+\ket\Psi|
 =|\psi_{\rm tr}^T\psi_{\rm tr}|,\\
 q(\Psi)&:=C_+(\Psi)^2
 =|\psi_{\rm tr}^T\psi_{\rm tr}|^2.
 \end{aligned}
 \label{eq:concurrence}
\end{equation}
Two normalized pure states belong to the same projective FLO orbit if and only
if they have the same value of \(C_+\) \cite{OszmaniecGuttKus2014}.  Hence the
fixed-concurrence families are exactly the fixed-\(q\) orbits studied in the
preceding sections.  We choose the pair \((\tau,T_{\rm tr})\) so that
\(T_{\rm tr}\ket{0000}_{\rm F}
=(\ket{1}_{\rm O}+i\ket{2}_{\rm O})/\sqrt2\).
This is possible by composing \(\tau\) with an inner automorphism and
\(T_{\rm tr}\) with the corresponding \(SO(8)\) matrix; for fixed
\(\tau\), the unitary intertwiner is unique up to a phase.
Thus the Fock vacuum is a complex superposition of the first two vectors of
the induced orthogonal basis, rather than either basis vector.  The pure
Gaussian states are exactly the quadric orbit
\(q=0\) \cite{OszmaniecGuttKus2014}.  At the other endpoint, \(q=1\), every
ray has a representative fixed by \(\theta_+\); equivalently, it has real
coordinates in \(\mathcal B_{\rm O}\).

\subsection{Two-copy Gaussian subspace}
\label{sec:gaussian-span}

Let \(\mathcal G_{4,+}\) be the normalized pure Gaussian vectors in
\(\mathcal H_+\).  Their two-copy span,
\begin{equation}
 \mathcal K_{4,+}:=
 \operatorname{span}\{\ket G^{\otimes2}:\ket G\in\mathcal G_{4,+}\}
 \subset\Sym^2(\mathcal H_+),
 \label{eq:four-mode-Gaussian-span}
\end{equation}
is the Gaussian-symmetric subspace \cite{deMeloCwikTerhal2013}.
It has dimension \(35\), whereas \(\Sym^2(\mathcal H_+)\) has dimension
\(36\).  For one, two, or three modes, the corresponding span fills the
symmetric square.  Appendix~\ref{app:gaussian-dimension} gives the
dimension calculation for \(m\) modes.

Let \(\Pi_G\) project onto \(T_{\rm tr}^{\otimes2}\mathcal K_{4,+}\).
Every Gaussian vector satisfies
\(\langle\Omega_8|\psi_{\rm tr}^{\otimes2}\rangle
=(\psi_{\rm tr}^T\psi_{\rm tr})/\sqrt8=0\).
Thus the transformed span lies in the traceless symmetric square.
Both spaces have dimension \(35\), so they coincide:
\begin{equation}
 \Pi_G=\frac{I+s}{2}-\proj{\Omega_8}=\Pi_0^{(8)}.
 \label{eq:PiG}
\end{equation}
Equation~\eqref{eq:null-moment} then gives
\(M_2^{(8)}(0)=\Pi_G/35\).

\subsection{The eight-dimensional cloning region}

\begin{corollary}[Four-mode concurrence orbits]
\label{cor:four-mode}
Consider arbitrary CPTP maps
\(\End(\mathcal H_+)\to\End(\mathcal H_+\otimes\mathcal H_+)\); the allowed
cloning operations are not restricted to FLO.  For the even sector of four
fermionic modes, set
\begin{equation}
 d=8,
 \qquad N_8=70,
 \qquad A_8=8-q,
 \qquad C_8=9q-2.
\end{equation}
Theorems~\ref{thm:moment}, \ref{thm:support}, and \ref{thm:kraus} give the
complete asymmetric local-cloning region for every \(q=C_+^2\in[0,1]\).
For each nonzero nonnegative supporting direction, they provide an optimal
channel of Choi rank eight.
In particular,
\begin{equation}
 M_2^{(8)}(q)=\frac{(8-q)(I+s)+(9q-2)e_8}{560},
 \label{eq:M2-eight}
\end{equation}
\begin{equation}
 F_{\rm sym}^{(8)}(q)
 =\frac{51-2q+\sqrt{144q^2+76q+1201}}{140},
 \label{eq:Fsym-eight}
\end{equation}
and the componentwise-minimal pair is
\begin{equation}
 (F_A,F_B)_{\rm min}^{(8,q)}=
 \begin{cases}
 \left(\dfrac{3+4q}{35},\dfrac{3+4q}{35}\right),
 &0\leq q\leq2/9,\\[2mm]
 \left(\dfrac{1-q}{7},\dfrac{1-q}{7}\right),
 &2/9\leq q\leq1.
 \end{cases}
 \label{eq:minimum-eight}
\end{equation}
\end{corollary}

\begin{figure}[!tbp]
 \includegraphics[width=\columnwidth]{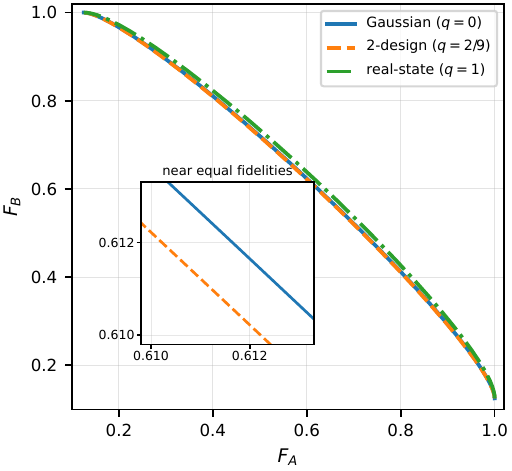}
 \caption{Pareto boundaries in the \((F_A,F_B)\) plane for three \(d=8\)
 orthogonal orbits, corresponding through triality to three four-mode
 fermionic concurrence classes.  The solid curve is the Gaussian orbit
 \(q=0\); the dashed curve is the projective \(2\)-design orbit \(q=2/9\),
 with the universal eight-dimensional average local-fidelity boundary.
 The dash-dotted curve is \(q=1\), the real-state ensemble in orthogonal
 coordinates.  The inset uses the same axes and resolves the Gaussian
 and \(2\)-design curves near equal fidelities.}
 \label{fig:qboundaries}
\end{figure}

All expressions in Corollary~\ref{cor:four-mode}, including \(e_8\) and the
transposes in the channel formulas, refer to the orthogonal coordinates
fixed by \(T_{\rm tr}\).  If \(\Phi_c\) is the coordinate-space channel of
Theorem~\ref{thm:kraus}, its action on physical even-parity registers is
\begin{equation}
 \Phi_c^{({\rm F})}(\rho)
 =(T_{\rm tr}^\dagger)^{\otimes2}
 \Phi_c\!\left(T_{\rm tr}\rho T_{\rm tr}^\dagger\right)
 T_{\rm tr}^{\otimes2}.
 \label{eq:physical-four-mode-channel}
\end{equation}

The \(q=2/9\) region coincides with the universal eight-dimensional
\(1\to2\) average local-fidelity region, whereas \(q=1\) coincides with the
real-state cloning region.  At the Gaussian point,
\begin{equation}
 F_{\rm sym}^{G}:=F_{\rm sym}^{(8)}(0)
 =\frac{51+\sqrt{1201}}{140}
 =0.6118246207\ldots,
 \label{eq:FG}
\end{equation}
which exceeds the universal value \(11/18\) by
\begin{align}
 F_{\rm sym}^{G}-\frac{11}{18}
 &=\frac{-311+9\sqrt{1201}}{1260}\nonumber\\
 &=7.135096\ldots\times10^{-4}.
 \label{eq:small-gap}
\end{align}
The trace distance between the quadric-orbit second moment and the Haar
second moment in
Eq.~\eqref{eq:null-Haar-distance} is \(1/36\), and their frame potentials are
\(1/35\) and \(1/36\), respectively.  The endpoint \(q=1\) has symmetric
fidelity \((7+\sqrt{29})/20\).

Figures~\ref{fig:qboundaries} and \ref{fig:Fq} show the Pareto boundaries for
three concurrence classes and the symmetric fidelity as a function of
\(q\).  Appendix~\ref{app:sextic} eliminates the support parameter for the
\(d=8\) Gaussian orbit and derives an implicit algebraic equation for its
boundary.

\begin{figure}[!tbp]
 \includegraphics[width=\columnwidth]{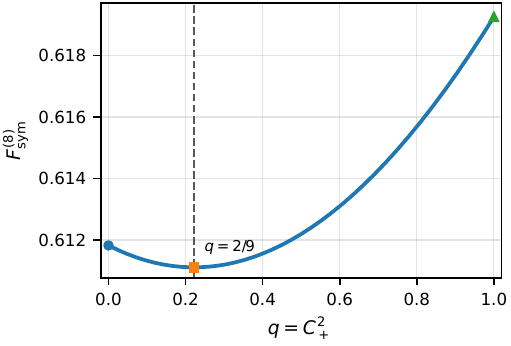}
 \caption{Optimal symmetric local fidelity for the four-mode family, with
 \(q=C_+^2\).  The circle, square, and triangle mark the Gaussian orbit
 \(q=0\), the projective \(2\)-design orbit \(q=2/9\), and the
 maximal-concurrence orbit \(q=1\), respectively.  The dashed vertical line
 locates the unique minimum at \(q=2/9\), where the second moment equals
 the Haar moment.}
 \label{fig:Fq}
\end{figure}

\section{Local and global cloning on the quadric orbit}
\label{sec:null-local-global}

\subsection{Unrestricted global optimum in arbitrary dimension}

For the quadric orbit, recall the projector \(\Pi_0^{(d)}\) onto the
traceless symmetric square defined in Eq.~\eqref{eq:Pi0}.
Its rank is \(\Nd/2\).
Using the state fidelity \(\mathsf F\) defined in
Eq.~\eqref{eq:state-fidelity}, define the average global two-copy fidelity
\begin{align}
 F_{\rm glob}^{(d)}(\Phi)
 &=\int_{\mathcal O_0^{(d)}}
 \mathsf F\!\left(\rho_\psi^{\otimes2},\Phi(\rho_\psi)\right)
 \,d\mu_0(\psi)\nonumber\\
 &=\int_{\mathcal O_0^{(d)}}
 \bra{\psi}^{\otimes2}\Phi(\rho_\psi)\ket{\psi}^{\otimes2}
 \,d\mu_0(\psi).
 \label{eq:global-fidelity}
\end{align}
This global criterion compares the full joint output with the two-copy
target, whereas
Eqs.~\eqref{eq:FA} and \eqref{eq:FB} compare the individual marginals.  The
same distinction occurs in mixed-state broadcasting
\cite{Barnum1996,LemmWilde2017,FanizzaEtAl2026}.

At \(q=0\), the fiducial vector in Eq.~\eqref{eq:orbit-normal-form} is a
highest-weight vector of the \(SO(d)\) vector representation.  The orbit
generated from such a vector is a family of generalized coherent states
\cite{Perelomov1986}.  Proposition~3 and the
uniform-prior remark in Ref.~\cite{ChiribellaYang2014} apply.  The one- and
two-copy orbit spans have dimensions \(d\) and \(\Nd/2\), respectively, so
the optimal coherent-state cloner yields the following \(1\to2\) result.

\begin{theorem}[Unrestricted quadric-orbit global optimum]
\label{thm:global-optimum}
For every \(d\geq3\) and every CPTP map
\(\Phi:\End(\C^d)\to\End(\C^d\otimes\C^d)\),
\begin{equation}
 F_{\rm glob}^{(d)}(\Phi)\leq\frac{2d}{\Nd}.
 \label{eq:global-optimum}
\end{equation}
Equality is attained by
\begin{equation}
 \Phi_{\rm glob}^{(d)}(\rho)
 =\frac{2d}{\Nd}\Pi_0^{(d)}(\rho\otimes I)\Pi_0^{(d)}.
 \label{eq:global-channel}
\end{equation}
\end{theorem}

\begin{proof}
The Choi objective operator associated with the global fidelity is
\begin{equation}
 Q_0^{(d)}=\int_{\mathcal O_0^{(d)}}
 \rho_\psi^{\otimes2}\otimes\rho_\psi^T\,d\mu_0(\psi).
\end{equation}
Since \(0\leq\rho_\psi^T\leq I\),
\begin{equation}
 0\leq Q_0^{(d)}
 \leq M_2^{(d)}(0)\otimes I
 =\frac{2}{\Nd}\Pi_0^{(d)}\otimes I.
\end{equation}
For every feasible Choi matrix, \(\Tr J_\Phi=d\), and therefore
\begin{align}
 F_{\rm glob}^{(d)}(\Phi)
 &=\Tr[J_\Phi Q_0^{(d)}]
 \leq\frac{2}{\Nd}
 \Tr\!\left[J_\Phi(\Pi_0^{(d)}\otimes I)\right]\nonumber\\
 &\leq\frac{2}{\Nd}\Tr J_\Phi
 =\frac{2d}{\Nd}.
\end{align}
Here the second inequality uses \(0\leq\Pi_0^{(d)}\leq I\).
The map in Eq.~\eqref{eq:global-channel} is completely positive.  Moreover,
\begin{equation}
 \Tr_B\Pi_0^{(d)}
 =\frac{d+1}{2}I-\frac1dI
 =\frac{\Nd}{2d}I,
\end{equation}
so it is trace preserving.  Finally,
\(\ket{\psi}^{\otimes2}\in\operatorname{ran}\Pi_0^{(d)}\) for every
quadric-orbit state, and hence
\begin{equation}
 \bra{\psi}^{\otimes2}
 \Phi_{\rm glob}^{(d)}(\rho_\psi)
 \ket{\psi}^{\otimes2}
 =\frac{2d}{\Nd}
 \bra{\psi}^{\otimes2}(\rho_\psi\otimes I)
 \ket{\psi}^{\otimes2}
 =\frac{2d}{\Nd}.
\end{equation}
The channel therefore attains the bound for every quadric-orbit input, and
hence on average.
\end{proof}

In the Kraus family of Theorem~\ref{thm:kraus}, the global channel corresponds to
\(c_{\rm glob}^{(d)}=(d,d,-2)^T/\sqrt{2\Nd}\).  Indeed,
\((d,d,-2)G_d(d,d,-2)^T=2\Nd\), and the Kraus operators satisfy
\(K_k(c_{\rm glob}^{(d)})
=\sqrt{2d/\Nd}\,\Pi_0^{(d)}(I\otimes\ket{k})\).
Summing \(K_k\rho K_k^\dagger\) over \(k\) recovers
Eq.~\eqref{eq:global-channel}.  Since \(W_{c_{\rm glob}^{(d)}}\) is an
isometry, this channel has Choi rank \(d\).

\subsection{Local--global fidelities in a symmetric channel family}

Consider the subfamily \(c=(x,x,z)^T\), with \(x,z\in\R\), of the channels
in Theorem~\ref{thm:kraus}, each of Choi rank \(d\).
Since exchanging the two outputs swaps
the first two components of \(c\), these channels are output-exchange
symmetric.  The subfamily contains the globally optimal channel in
Eq.~\eqref{eq:global-channel} and the symmetric local optimizer derived
below.  Set
\(F_{\rm glob}^{(d)}(c):=F_{\rm glob}^{(d)}(\Phi_c)\).  On the quadric orbit,
\(g:=F_{\rm glob}^{(d)}(c)=4x^2/d\), while the common single-clone fidelity
\begin{equation}
 \begin{aligned}
 F_{\rm loc}^{(d)}(c)
 &:=F_A^{(d,0)}(\Phi_c)=F_B^{(d,0)}(\Phi_c)\\
 &=x^2\left(1+\frac3d\right)
 +\frac{2xz}{d}+\frac{z^2}{d}.
 \end{aligned}
 \label{eq:local-null}
\end{equation}
The normalization \(c^TG_dc=1\) becomes
\(z^2+4xz/d+2(d+1)x^2/d=1\).
Completing the square gives
\begin{equation}
 \left(z+\frac{2x}{d}\right)^2
 =1-\frac{\Nd}{2d}g,
 \qquad
 0\leq g\leq\frac{2d}{\Nd}.
\end{equation}
Because \(c\) and \(-c\) determine the same channel, we may take
\(x\geq0\).  The two signs below then correspond to the two solutions
\begin{equation}
 x=\frac12\sqrt{dg},
 \qquad
 z+\frac{2x}{d}
 =\pm\sqrt{1-\frac{\Nd}{2d}g},
 \label{eq:local-global-signs}
\end{equation}
of the normalization condition.
The upper endpoint is attained if and only if
\begin{equation}
 z=-\frac{2x}{d},
 \qquad\text{equivalently}\qquad
 \frac{x}{z}=-\frac d2.
 \label{eq:global-equality-condition}
\end{equation}
Eliminating \(x\) and \(z\) yields the two branches
\begin{multline}
 F_{\rm loc}^{(d),\pm}(g)
 =\frac1d
 +\frac{d^3+d^2-6d+8}{4d^2}\,g\\
 {}\pm\frac{d-2}{d^2}
 \sqrt{g\left(d-\frac{\Nd}{2}g\right)},
 \qquad
 0\leq g\leq\frac{2d}{\Nd}.
 \label{eq:local-global-curve}
\end{multline}
Equation~\eqref{eq:local-global-curve} gives the exact locus of the family
\(c=(x,x,z)^T\), with real coefficients.  It does not describe all
exchange-symmetric channels of Choi rank \(d\), or the unrestricted
local--global boundary.  The globally optimal channel lies at the common
endpoint \(g=2d/\Nd\).  Its local
fidelity is
\begin{equation}
 F_{\rm loc}^{(d)}(c_{\rm glob}^{(d)})
 =\frac{d^3+3d^2-4d+4}{2d\Nd}.
 \label{eq:local-of-global}
\end{equation}

\subsection{Unique symmetric local optimizer}

Put
\begin{equation}
 \mathcal D_d=\Nd^2-16(d-2),
 \qquad
 \kappa_d=\frac{\Nd+\sqrt{\mathcal D_d}}8.
 \label{eq:kappa-d}
\end{equation}
The symmetric quadric-orbit local optimizer has multiplicity vector
\begin{equation}
 c_*^{(d)}=
 \frac{(\kappa_d,\kappa_d,-1)^T}
 {\sqrt{(\kappa_d,\kappa_d,-1)G_d
 (\kappa_d,\kappa_d,-1)^T}}.
 \label{eq:cstar-general}
\end{equation}
Indeed,
\begin{equation}
 B_0^{(d)}(1,1)(\kappa_d,\kappa_d,-1)^T
 =(2d+4\kappa_d)(\kappa_d,\kappa_d,-1)^T,
\end{equation}
where \(2d+4\kappa_d=(\Nd+4d+\sqrt{\mathcal D_d})/2\) is the largest
eigenvalue on the vector isotypic component, as established after
Eqs.~\eqref{eq:symmetric-multiplicity-block} and
\eqref{eq:symmetric-sector-dominance}.  Consequently,
Theorem~\ref{thm:kraus} and Eq.~\eqref{eq:symmetric-fidelity-definition} give
\begin{equation}
 F_A^{(d,0)}(c_*^{(d)})=F_B^{(d,0)}(c_*^{(d)})
 =\frac{2d+4\kappa_d}{2\Nd}
 =F_{\rm sym}^{(d)}(0).
\end{equation}
Its ratio is \(x/z=-\kappa_d\),
whereas the globally optimal vector has \(x/z=-d/2\).  For \(d\geq4\),
\(\Nd>4d\) implies \(\kappa_d>d/2\); for \(d=3\), the same conclusion
follows from \(\sqrt{\mathcal D_d}>4d-\Nd=2\).  Thus the two channels are
distinct, and the equality condition in
Eq.~\eqref{eq:global-equality-condition} gives
\(F_{\rm glob}^{(d)}(c_*^{(d)})<2d/\Nd\).

\begin{theorem}[Uniqueness and Choi rank]
\label{thm:unique}
For every \(d\geq3\), the channel \(\Phi_{c_*^{(d)}}\) is the unique CPTP
map maximizing the equally weighted average local objective
\begin{equation}
 \frac12\left[F_A^{(d,0)}(\Phi)+F_B^{(d,0)}(\Phi)\right].
\end{equation}
Its Choi matrix is
\begin{equation}
 J_*^{(d)}=W_{c_*^{(d)}}W_{c_*^{(d)}}^\dagger
\end{equation}
and \(\rank J_*^{(d)}=d\).
Consequently its minimal Kraus number and minimal Stinespring environment
dimension are both \(d\) \cite{Choi1975,Stinespring1955}.
\end{theorem}

Appendix~\ref{app:unique} proves uniqueness without assuming that competing
optimizers are covariant.

\subsection{Four-mode specialization}

At \(d=8\), the quadric orbit is the pure four-mode Gaussian orbit in the
even-parity sector \(\mathcal H_+\).  The general formulas reduce to
\begin{equation}
 \max_{\Phi\ \mathrm{CPTP}}F_{\rm glob}^{(8)}(\Phi)=\frac8{35},
 \qquad
 c_{\rm glob}^{(8)}=\frac1{\sqrt{35}}(4,4,-1)^T.
 \label{eq:cglob-eight}
\end{equation}
The corresponding local fidelity is
\(F_{\rm loc}^{(8)}(c_{\rm glob}^{(8)})=169/280\), and
\begin{equation}
 \begin{gathered}
 F_{\rm loc}^{(8),\pm}(g)
 =\frac18+\frac{67}{32}g
 \pm\frac3{32}\sqrt{g(8-35g)},\\
 0\leq g\leq\frac8{35}.
 \end{gathered}
 \label{eq:local-global-eight}
\end{equation}
Figure~\ref{fig:localglobal} plots this restricted locus; its inset separates
the two nearby optima.
\begin{figure}[!t]
 \centering
 \includegraphics[width=\columnwidth]{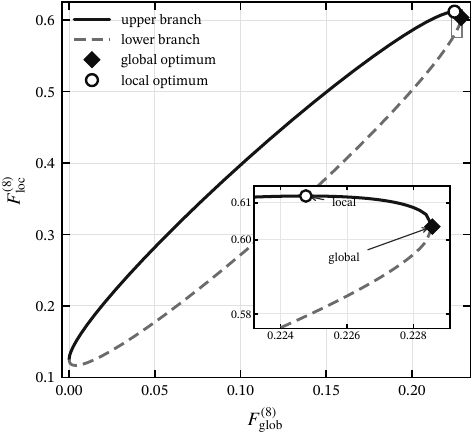}
 \caption{Average local fidelity versus average global fidelity on the
 four-mode Gaussian orbit, restricted to the covariant channels
 \(\Phi_c\) with real coefficients \(c=(x,x,z)^T\) and \(c^TG_8c=1\).
 These channels are output-exchange symmetric and have Choi rank eight.
 The solid and dashed curves are the \(+\) and \(-\) branches of
 Eq.~\eqref{eq:local-global-eight}.  The diamond at their common right
 endpoint marks the global optimum; the circle marks the unique local
 optimum.  The inset separates these points.  The plotted curves describe only this family of channels and not the full
set of attainable local--global fidelity pairs.}
 \label{fig:localglobal}
\end{figure}

At \(d=8\), \(\kappa_8=(35+\sqrt{1201})/4\), so the locally optimal vector has
\(x/z=-(35+\sqrt{1201})/4\).  Its local and global fidelities are
\begin{align}
 F_{\rm loc}^{(8)}(c_*^{(8)})
 &=\frac{51+\sqrt{1201}}{140},\nonumber\\
 F_{\rm glob}^{(8)}(c_*^{(8)})
 &=\frac{2(2402+67\sqrt{1201})}{42035}\nonumber\\
 &=0.2247610297\ldots<\frac8{35}.
 \label{eq:global-of-local-eight}
\end{align}
The value \(8/35\) is the four-mode, one-parity-sector specialization of
the coherent-orbit formula presented in Ref.~\cite{Herasymenko2026}.

\subsection{Obstruction to Gaussian implementation}

For the physical eight-mode output, order the four modes of \(A\) before
those of \(B\).  The fixed-even-parity output registers then have the
ordinary tensor-product description \(\mathcal H_+\otimes\mathcal H_+\)
used above.  We call a state convex-Gaussian if it is a
convex mixture of pure fermionic Gaussian states
\cite{deMeloCwikTerhal2013,OszmaniecGuttKus2014}.

\begin{proposition}[Obstruction to Gaussian implementation]
\label{prop:gaussian-implementation}
On the physical output space \(\mathcal H_+\otimes\mathcal H_+\), the unique
symmetric local optimizer \(\Phi_{c_*^{(8)}}^{({\rm F})}\) for the four-mode
Gaussian orbit cannot be implemented by a protocol whose output on every
Gaussian input is a convex mixture of fermionic Gaussian states.  In
particular, it cannot be implemented by FLO with Gaussian ancillas and
partial trace.  The same obstruction applies to the particular globally
optimal channel \(\Phi_{c_{\rm glob}^{(8)}}^{({\rm F})}\) in
Eq.~\eqref{eq:cglob-eight}.
\end{proposition}

\begin{proof}
In a suitable Majorana basis related to the original one by an
\(SO(2m)\) transformation, an \(m\)-mode Gaussian state has the normal form
\cite{Bravyi2005}
\begin{equation}
 \sigma=2^{-m}\prod_{j=1}^m(I+\nu_j Z_j),
 \qquad
 Z_j=i\widetilde\gamma_{2j-1}\widetilde\gamma_{2j},
 \qquad |\nu_j|\leq1.
 \label{eq:gaussian-normal-form}
\end{equation}
The operators \(Z_j\) commute and square to the identity.  In this
\(SO(2m)\)-related Majorana basis, the parity operator is
\(P_m=i^m\widetilde\gamma_1\cdots\widetilde\gamma_{2m}=\prod_jZ_j\).
Hence
\begin{equation}
 \Tr(\sigma P_m)=\prod_{j=1}^m\nu_j,
 \qquad
 \Tr(\sigma^2)=\prod_{j=1}^m\frac{1+\nu_j^2}{2}.
 \label{eq:gaussian-parity-purity}
\end{equation}
If \(\sigma\) has definite parity, then
\(|\Tr(\sigma P_m)|=1\), so every \(|\nu_j|=1\) and
Eq.~\eqref{eq:gaussian-parity-purity} shows that \(\sigma\) is pure.

Now let \(\sigma_{AB}\) be an eight-mode Gaussian state supported in
\(\mathcal H_+\otimes\mathcal H_+\).  Its four-mode marginals are Gaussian
\cite{Bravyi2005} and have definite even parity; they are therefore pure.
A bipartite state with a pure marginal is a product state, so
\(\sigma_{AB}=\proj{G_A}\otimes\proj{G_B}\) for pure even Gaussian states
\(G_A,G_B\).  More generally, if a convex mixture of Gaussian states is
supported in \(\mathcal H_+\otimes\mathcal H_+\), positivity implies that
every positive-weight constituent is supported there.  Thus every such
convex-Gaussian state is separable across \(A{:}B\).

With the triality convention fixed after Eq.~\eqref{eq:concurrence},
\(\psi_0=(\ket{1}_{\rm O}+i\ket{2}_{\rm O})/\sqrt2\) represents the Fock
vacuum and is therefore a Gaussian input.  For \(c=(x,x,z)^T\), set
\(\rho_c=\Phi_c(\rho_{\psi_0})\).  Direct substitution in
Eq.~\eqref{eq:general-kraus} gives, for distinct indices \(3\leq r,s\leq d\),
\begin{equation}
 \langle r,r|\rho_c|s,s\rangle=\frac{z^2}{d},
 \qquad
 \langle r,s|\rho_c|r,s\rangle
 =\langle s,r|\rho_c|s,r\rangle=0.
 \label{eq:gaussian-output-elements}
\end{equation}
Consequently, for \(\eta=(\ket{r,s}-\ket{s,r})/\sqrt2\),
\begin{equation}
 \langle\eta|\rho_c^{T_B}|\eta\rangle=-\frac{z^2}{d}<0
 \qquad (z\ne0).
 \label{eq:gaussian-npt-witness}
\end{equation}
Here \(T_B\) is the ordinary Hilbert-space partial transpose in
\(\mathcal B_{\rm O}\) on register \(B\), not a fermionic partial transpose.
At \(d=8\), one may take \(r=3\) and \(s=4\).
Equation~\eqref{eq:physical-four-mode-channel} conjugates the output locally
by \(T_{\rm tr}^\dagger\), which preserves the spectrum of the partial
transpose.  Every separable state has a positive partial
transpose, so the negative expectation in Eq.~\eqref{eq:gaussian-npt-witness}
certifies entanglement.  Both
\(c_*^{(8)}\) in Eq.~\eqref{eq:cstar-general} and
\(c_{\rm glob}^{(8)}\) in Eq.~\eqref{eq:cglob-eight} have \(z\ne0\), so
the corresponding physical outputs are entangled and hence are not
convex-Gaussian.  Finally, Theorem~\ref{thm:unique} rules out a different
Gaussian protocol attaining the same symmetric local optimum.
\end{proof}
We do not prove uniqueness of the global optimum here.  The above argument
only shows that the globally optimal channel in
Eq.~\eqref{eq:cglob-eight} is not Gaussian implementable.

\section{Discussion and conclusions}
\label{sec:conclusions}

In this paper, we have studied \(1\to2\) asymmetric cloning of pure states
belonging to orthogonal-group orbits.  The orbits are parametrized by
\(q=|\psi^T\psi|^2\), and we have determined the complete region of
achievable local fidelities for every \(q\in[0,1]\) and every \(d\geq3\).
Using orthogonal covariance and the Brauer algebra, the optimization over
cloning channels can be reduced to a finite-dimensional spectral problem.
What is more, for every nonzero nonnegative supporting direction we give an
explicit optimal cloner with Choi rank \(d\).

We have also studied the symmetric case in more detail.  The optimal
symmetric fidelity has a unique minimum at
\(q=2/(d+1)\), where the second moment of the orbit is equal to the Haar
second moment and the cloning problem becomes the universal one.  For
\(d=8\), triality identifies \(q\) with the squared fermionic concurrence
\cite{OszmaniecGuttKus2014}.  In this case \(q=0\) corresponds to pure
Gaussian states, while \(q=1\) corresponds to maximal concurrence.  The
minimum occurs between these two points, showing that the optimal symmetric
cloning fidelity is not monotone in fermionic concurrence.

For the Gaussian orbit we have additionally compared local and global
cloning.  The globally optimal channel is given by the coherent-state
cloner of Ref.~\cite{ChiribellaYang2014}, while the symmetric locally
optimal channel is different.  We prove that the local optimizer is unique
among all CPTP maps, without imposing covariance, and that its Choi rank is
\(d\).  We also calculate explicitly the local and global fidelities inside
the real exchange-symmetric family containing both channels.  Thus the
difference between the local and global optima is not an artifact of the
covariance reduction.

Our results concern the finite-copy \(1\to2\) problem.  They do not address
the high-fidelity regime in which many input copies are available.  In
particular, Fanizza \emph{et al.} conjecture that producing one additional
copy of an \(m\)-mode Gaussian state with joint fidelity at least
\(1-\varepsilon\) requires \(\Theta(m^2/\varepsilon)\) input copies, up to
constant factors \cite{FanizzaEtAl2026}.

There are several natural questions left open.  One is the extension of the
present method to \(1\to n\) cloning on the same family of orthogonal
orbits.  Another is to determine the optimal fidelity region when the
cloning operations themselves are restricted to fermionic Gaussian
operations.  We have shown that such operations cannot attain the symmetric
local optimum for the four-mode Gaussian orbit.  It would also be
interesting to understand the situation for more than four modes.  In this
case triality is no longer available and one has to deal directly with
centralizer algebras of spin representations \cite{Wenzl2012}.

\section*{Acknowledgments}
This work was carried out under the IRA Programme, project
No.~FENG.02.01-IP.05-0006/23, financed by the FENG Programme 2021--2027,
Priority FENG.02, Measure FENG.02.01, with support from the FNP. ChatGPT
(OpenAI) was used for language editing and for discussing the
presentation and exposition of the manuscript.

\appendix

\section{Three-copy block decomposition and Kraus normalization}
\label{app:spectrum}

This appendix derives the operator decomposition used in
Theorem~\ref{thm:support}.  Put \(s_1=a+b\) and
\(r=\sqrt{a^2+b^2-ab}\).  For the embeddings in
Eq.~\eqref{eq:Tmaps}, direct contraction gives
\(T_\alpha^\dagger T_\beta=(G_d)_{\alpha\beta}I\), with \(G_d\) in
Eq.~\eqref{eq:general-G}.  In the ordered multiplicity basis
\((T_1,T_2,T_3)\), the swaps act as follows, where \(\widehat{=}\) denotes
the corresponding representation matrix:
\begin{equation}
 s_{13}\ \widehat{=}\
 \begin{pmatrix}0&0&1\\0&1&0\\1&0&0\end{pmatrix},
 \qquad
 s_{23}\ \widehat{=}\
 \begin{pmatrix}1&0&0\\0&0&1\\0&1&0\end{pmatrix},
 \label{eq:swap-actions}
\end{equation}
and the contractions act as
\begin{equation}
 e_{13}\ \widehat{=}\
 \begin{pmatrix}0&0&0\\1&d&1\\0&0&0\end{pmatrix},
 \qquad
 e_{23}\ \widehat{=}\
 \begin{pmatrix}d&1&1\\0&0&0\\0&0&0\end{pmatrix}.
 \label{eq:contraction-actions}
\end{equation}
Substitution into Eq.~\eqref{eq:general-K} yields the multiplicity matrix in Eq.~\eqref{eq:general-B}.  Its characteristic polynomial is
\begin{align}
 \det(\mu I-B_q^{(d)})={}&
 \mu^3-(2\Ad+\Nd)s_1\mu^2\nonumber\\
 &+(\Ad^2+2\Ad\Nd-\Cd^2)s_1^2\mu\nonumber\\
 &+\bigl(\Nd^2-2\Ad\Cd-2\Ad\Nd
 +2\Cd^2\bigr)ab\,\mu\nonumber\\
 &+\Nd(\Cd^2-\Ad^2)s_1^3\nonumber\\
 &+\Nd\bigl(2\Ad^2+\Ad\Cd\nonumber\\
 &\hspace{2.8em}{}-\Ad\Nd-3\Cd^2\bigr)ab\,s_1.
 \label{eq:general-cubic}
\end{align}
The relation \((B_q^{(d)})^TG_d=G_dB_q^{(d)}\) follows either by direct
multiplication or from Eq.~\eqref{eq:pencil-identity}, because
\(G_dB_q^{(d)}=\Nd[aM_A^{(d)}+bM_B^{(d)}]\) is real symmetric.

On \(\mathcal T_d^\perp\), \(e_{13}=e_{23}=0\), and
\(d\Nd K_q^{(d)}=\Ad(a+b)I+\Cd(as_{13}+bs_{23})\).
The trivial and sign representations give \((\Ad+\Cd)s_1\) and \((\Ad-\Cd)s_1\), respectively.  On the standard representation one may take
\begin{equation}
 s_{13}=\begin{pmatrix}1&0\\0&-1\end{pmatrix},
 \qquad
 s_{23}=\begin{pmatrix}-1/2&\sqrt3/2\\\sqrt3/2&1/2\end{pmatrix}.
\end{equation}
The matrix \(as_{13}+bs_{23}\) has trace zero and determinant
\(-(a^2+b^2-ab)\), proving the four complementary branches in
Eq.~\eqref{eq:general-support}.  Their multiplicities, together with those
of the eigenvalues on the vector isotypic component, are listed in
Table~\ref{tab:spectrum-multiplicities}.  In particular, every branch in
Eq.~\eqref{eq:general-support} occurs for every \(d\geq3\).

\begin{table}[!b]
 \caption{Eigenvalue multiplicities in the spectrum of
 \(d\Nd K_q^{(d)}(a,b)\).}
 \label{tab:spectrum-multiplicities}
 \footnotesize
 \setlength{\tabcolsep}{3pt}
 \renewcommand{\arraystretch}{1.15}
 \begin{tabular}{@{}p{0.29\columnwidth}ll@{}}
 \toprule
 \raggedright \(O(d)\) component & Eigenvalue & Multiplicity \\
 \midrule
 \raggedright Vector isotypic component
  & \(\mu_1,\mu_2,\mu_3\) & \(d\) each \\
 \raggedright Symmetric traceless tensors of order three
  & \((\Ad+\Cd)s_1\) & \(d(d-1)(d+4)/6\) \\
 \raggedright Alternating tensors of order three
  & \((\Ad-\Cd)s_1\) & \(d(d-1)(d-2)/6\) \\
 \raggedright Traceless mixed-symmetry tensors
  & \(\Ad s_1\mathbin{\pm}\Cd r\) & \(d(d^2-4)/3\) \\
 \bottomrule
 \end{tabular}
\end{table}
Here mixed symmetry denotes the part orthogonal to the fully symmetric
and fully alternating tensors; tracelessness additionally requires that
all pair contractions vanish, as in Eq.~\eqref{eq:traceless-sector}.
At coincident eigenvalues the corresponding multiplicities add, and the
dimension check is
\begin{equation}
 3d+\frac{d(d-1)(d+4)}6+\frac{d(d-1)(d-2)}6
 +\frac{2d(d^2-4)}3=d^3.
 \label{eq:spectrum-dimension-check}
\end{equation}

\subsection{Kraus completeness calculation}

We verify the trace-preservation condition for the Kraus
operators in Eq.~\eqref{eq:general-kraus}.  Define
\begin{equation}
 \begin{aligned}
 A_k&:=\sum_{j=1}^d\ket{k,j}\bra j,\\
 B_k&:=\sum_{j=1}^d\ket{j,k}\bra j,\\
 C_k&:=\sum_{j=1}^d\ket{j,j}\bra k,\\
 K_k&=\frac1{\sqrt d}(c_1A_k+c_2B_k+c_3C_k).
 \end{aligned}
 \label{eq:kraus-building-blocks}
\end{equation}
Orthonormality of the fixed basis \(\mathcal B_{\rm O}\) gives
\begin{equation}
 \begin{aligned}
 A_k^\dagger A_k&=B_k^\dagger B_k=I,
 &C_k^\dagger C_k&=d\proj{k},\\
 A_k^\dagger B_k&=A_k^\dagger C_k=B_k^\dagger C_k=\proj{k},
 \end{aligned}
 \label{eq:kraus-block-products}
\end{equation}
and the adjoints of the mixed products in the second line give the same
rank-one projector.  Since \(c_1,c_2,c_3\) are real,
\begin{equation}
 K_k^\dagger K_k
 =\frac{c_1^2+c_2^2}{d}I
 +\left[c_3^2+\frac2d(c_1c_2+c_1c_3+c_2c_3)\right]\proj{k}.
 \label{eq:single-kraus-product}
\end{equation}
Summing over \(k\) and using \(\sum_k\proj{k}=I\) therefore yields
\begin{align}
 \sum_{k=1}^dK_k^\dagger K_k
 &=\left[c_1^2+c_2^2+c_3^2
 +\frac2d(c_1c_2+c_1c_3+c_2c_3)\right]I\nonumber\\
 &=(c^TG_dc)I=I.
 \label{eq:kraus-completeness-calculation}
\end{align}

\section{Marginal fidelity bounds and Brauer-state extensions}
\label{app:marginal-extrema}

\subsection{Minimum marginal fidelity}

To determine the minimum attainable fidelity of either marginal, use the
decomposition
\(\C^d\otimes\C^d=\C\Omega_d\oplus\Sym_0^2(\C^d)
\oplus\wedge^2(\C^d)\), where \(\wedge^2(\C^d)\) is the antisymmetric
subspace.  On these three summands, \(R_q^{(d)}\) has
eigenvalues \(1/d\), \((\Ad+\Cd)/(d\Nd)\), and
\((\Ad-\Cd)/(d\Nd)\), respectively.  For a general one-to-two
channel \(\Phi\), write \(\Phi_A:=\Tr_B\circ\Phi\) and
\(\Phi_B:=\Tr_A\circ\Phi\).  The reduced Choi matrix
\(J_{\Phi_A}=\Tr_BJ_\Phi\) is positive and satisfies
\(\Tr_AJ_{\Phi_A}=I_C\), hence \(\Tr J_{\Phi_A}=d\).  It follows from
Eq.~\eqref{eq:fidelities-Choi-R} that
\(F_A^{(d,q)}(\Phi)=\Tr[J_{\Phi_A}R_q^{(d)}]
\geq d\lambda_{\min}[R_q^{(d)}]\).
The bound is attainable.  If \(P\) is an invariant spectral projector of
\(R_q^{(d)}\), with \(r=\rank P\), then Schur's lemma and normalization give
\(\Tr_A P=(r/d)I_C\).  Hence \(J_P=(d/r)P\) is the Choi matrix of a
one-output channel and attains \(d\lambda\) on an eigenspace with eigenvalue
\(\lambda\).  Appending the maximally mixed state \(I/d\) as the second
output produces a one-to-two channel with the same marginal fidelity.
Consequently,
\begin{equation}
 \begin{aligned}
 F_{\min}^{(d)}(q)
 &=d\,\lambda_{\min}\!\left[R_q^{(d)}\right]\\
 &=\min\!\left\{1,
 \frac{\Ad+\Cd}{\Nd},
 \frac{\Ad-\Cd}{\Nd}\right\}.
 \end{aligned}
 \label{eq:marginal-minimum-spectral-form}
\end{equation}
For \(d\geq3\) and \(q\in[0,1]\), the last two candidates are
nonnegative and strictly smaller than \(1\), so the invariant-line
eigenvalue is never minimal.  Since
\(\min\{\Ad+\Cd,\Ad-\Cd\}=\Ad-|\Cd|\), and \(\Cd\) changes sign at
\(q=2/(d+1)\), this gives Eq.~\eqref{eq:marginal-minimum}.

The individual lower bounds imply
\(h_q^{(d)}(-1,-1)\leq-2F_{\min}^{(d)}(q)\).
For \(q\leq q_{\rm H}\), the trivial branch of
Theorem~\ref{thm:support}, evaluated at \(a=b=-1\), attains this bound;
for \(q\geq q_{\rm H}\), the sign branch does so.  Hence equality holds.
An optimizer of
\(-F_A-F_B\) therefore satisfies
\(F_A+F_B=2F_{\min}^{(d)}(q)\).  Since both marginal fidelities are at
least \(F_{\min}^{(d)}(q)\), both equal the minimum, proving Eq.~\eqref{eq:componentwise-minimum}.

\subsection{Perfect-copy endpoints}

Suppose that \(F_A^{(d,q)}(\Phi)=1\).  Since
\(1-f_A(\psi;\Phi)\) is continuous and nonnegative and the invariant
measure has full support, the first marginal equals \(\rho_\psi\) for every
orbit input.  Take a Stinespring isometry
\(V:\mathcal H\to\mathcal H_A\otimes\mathcal H_B\otimes\mathcal H_E\),
where \(E\) is an auxiliary environment and
\(\Phi(X)=\Tr_E(VXV^\dagger)\) \cite{Stinespring1955}.
Purity of the first marginal gives
\(V\ket\psi=\ket\psi_A\otimes\ket{\eta_\psi}_{BE}\).  For orbit vectors
with \(\langle\psi|\phi\rangle\neq0\), preservation of inner products
implies \(\langle\eta_\psi|\eta_\phi\rangle=1\).  The projective orbit is
connected, so a path subdivided into nonorthogonal consecutive rays shows
that \(\eta_\psi\) is independent of \(\psi\).  Since the orbit spans
\(\C^d\), linearity yields \(\Phi(X)=X\otimes\sigma_B\) for a fixed state
\(\sigma_B\).  The \(1\)-design identity then gives
\(F_B^{(d,q)}(\Phi)=1/d\).  Conversely, appending any fixed state to an
unchanged input attains this pair.  Exchanging the outputs gives the other
endpoint.

\subsection{Symmetric extensions of Brauer states}
\label{app:brauer-extensions}

We assign extendibility indices using the bipartition \(C{:}A\), with the
extension taken on the output factor \(A\), although tensor factors are
displayed as \(AC\).  Thus \(\sigma_{AC}\) is \((1,2)\)-extendible if there
exists a state \(\widetilde\sigma_{ABC}\), with \(B\simeq A\), that is
invariant under \(A\leftrightarrow B\) and whose \(AC\) and \(BC\) marginals
both equal \(\sigma\) after identifying \(B\simeq A\).

\begin{proposition}
\label{prop:brauer-extendibility}
Under this convention, a state \(\sigma_{AC}\) occurs as the common
\(AC\) and \(BC\) marginal of the normalized Choi state of an
exchange-symmetric \(O(d)\)-covariant one-to-two channel if and only if it
is a \((1,2)\)-extendible Brauer state of local dimension \(d\).
\end{proposition}

\begin{proof}
Let \(\omega_{ABC}=J_\Phi/d\) be the normalized Choi state.
If the channel is exchange symmetric, then
\(\omega_{AC}=\omega_{BC}=: \sigma\), after identifying \(A\simeq B\), and
\(\omega_{ABC}\) is an exchange-invariant extension of \(\sigma\).
Conversely, take a \((1,2)\)-extendible Brauer state \(\sigma\) and an
exchange-invariant extension.
Twirling the extension by \(O^{\otimes3}\) preserves its two specified
marginals and produces an \(O(d)\)-invariant extension that remains
exchange invariant.  Its \(C\)
marginal commutes with the irreducible vector representation, hence equals
\(I/d\).  Multiplying by \(d\) gives the Choi matrix of an
exchange-symmetric covariant channel.
\end{proof}

\section{Gaussian-symmetric subspaces for arbitrary mode number}
\label{app:gaussian-dimension}

Consider the \(m\)-mode Fock space
\(\mathcal F_m=\mathcal H_{m,+}\oplus\mathcal H_{m,-}\) with Majoranas
\(\gamma_1,\ldots,\gamma_{2m}\) satisfying Eq.~\eqref{eq:CAR}, and let
\(\widehat N_m:=\sum_{j=1}^m a_j^\dagger a_j\).  Its parity operator and the
dimension of either parity sector are
\begin{equation}
 \begin{aligned}
 P_m&:=(-1)^{\widehat N_m}=i^m\gamma_1\cdots\gamma_{2m},\\
 d_m&:=\dim\mathcal H_{m,\pm}=2^{m-1}.
 \end{aligned}
 \label{eq:general-parity-dimension}
\end{equation}
Define the two-copy Majorana operator on
\(\mathcal F_m\otimes\mathcal F_m\) by
\begin{equation}
 \Lambda_m:=\sum_{r=1}^{2m}\gamma_r\otimes\gamma_r.
 \label{eq:Lambda}
\end{equation}
Here \(\otimes\) is the ordinary Hilbert-space tensor product of two
copies.  Each Majorana reverses parity, so \(\Lambda_m\) maps
\(\mathcal H_{m,\pm}^{\otimes2}\) into
\(\mathcal H_{m,\mp}^{\otimes2}\).  For a normalized
\(\ket\psi\in\mathcal H_{m,\pm}\),
\begin{equation}
 \ket\psi\text{ is Gaussian}
 \quad\Longleftrightarrow\quad
 \Lambda_m\ket\psi^{\otimes2}=0.
 \label{eq:gaussian-criterion}
\end{equation}
Let \(\mathcal G_{m,\pm}\subset\mathcal H_{m,\pm}\) be the set of normalized
pure Gaussian vectors.  More precisely, the fixed-parity kernels satisfy
\begin{equation}
 \begin{aligned}
 \mathcal K_{m,\pm}
 &:={\left\{v\in\mathcal H_{m,\pm}^{\otimes2}:\Lambda_m v=0\right\}}\\
 &=\operatorname{span}{\left\{
 \ket G^{\otimes2}:\ket G\in\mathcal G_{m,\pm}\right\}}\\
 &\subset\Sym^2(\mathcal H_{m,\pm}).
 \end{aligned}
 \label{eq:fixed-parity-Gaussian-kernel}
\end{equation}
The subspace \(\mathcal K_{m,\pm}\) is the Gaussian-symmetric subspace in
the fixed-parity sector \cite{deMeloCwikTerhal2013}.  The
operator \(\Lambda_m\) also enters quantum algorithms for estimating
non-Gaussianity and testing Gaussianity and the projective \(2\)-design property
\cite{Tarabunga2026}.

Reference~\cite{deMeloCwikTerhal2013} also gives the full-kernel dimension
\begin{equation}
 \dim\ker\Lambda_m=\binom{2m}{m}.
 \label{eq:full-Lambda-kernel-dimension}
\end{equation}
Its parity decomposition follows by setting
\(X_r:=\gamma_r\otimes\gamma_r\).  The operators \(X_r\) commute and have eigenvalues
\(\pm1\).  Thus \(\Lambda_m=0\) in their joint eigenbasis precisely when
\(m\) eigenvalues are \(+1\) and \(m\) are \(-1\).  Moreover,
\(P_m\otimes P_m=(-1)^m\prod_{r=1}^{2m}X_r\).
On a zero-eigenvector of \(\Lambda_m\), the product on the right-hand side is
\((-1)^m\).  Hence
\((P_m\otimes P_m)v=v\) for every \(v\in\ker\Lambda_m\), so the kernel contains
only the equal-parity sectors.  The unitary \(\gamma_1\otimes\gamma_1\) commutes with
\(\Lambda_m\) and maps \(\mathcal K_{m,+}\) bijectively onto
\(\mathcal K_{m,-}\).
It follows that the Gaussian-symmetric subspace in either fixed-parity
sector has dimension
\begin{equation}
 D_m:=\dim\mathcal K_{m,+}=\dim\mathcal K_{m,-}
 =\frac12\binom{2m}{m}.
\end{equation}
Its codimension inside the symmetric square is
\begin{equation}
 \Delta_m=\frac{d_m(d_m+1)}2-D_m.
 \label{eq:fermionic-defect}
\end{equation}
For the first few values,
\begin{equation}
\begin{array}{c|cccccc}
m&1&2&3&4&5&6\\ \hline
d_m&1&2&4&8&16&32\\
D_m&1&3&10&35&126&462\\
\Delta_m&0&0&0&1&10&66
\end{array}
\label{eq:defect-table}
\end{equation}
Thus four modes are the first case in which the Gaussian-symmetric subspace
is a proper subspace of the symmetric square, and its codimension is one.

At \(m=4\), this gives \(\dim\mathcal K_{4,+}=35\), as used in
Sec.~\ref{sec:gaussian-span}.

\section{Four-mode Gaussian boundary as an algebraic curve}
\label{app:sextic}

For \(d=8\) and \(q=0\), triality identifies the quadric orbit with the
four-mode fermionic Gaussian orbit.  Eliminating the support parameter
yields an implicit algebraic equation for its Pareto boundary.  Put
\(X=F_A\), \(Y=F_B\), \(u=X+Y\), and \(v=XY\).
For a Pareto supporting direction \(a,b\geq0\), not both zero, let
\(\mu=\mu_+(8,0;a,b)\) be the largest eigenvalue of the vector multiplicity
matrix, as in Corollary~\ref{cor:positive}, and write
\begin{equation}
 s_1=a+b,
 \qquad
 \tau=\frac{ab}{s_1^2},
 \qquad
 \lambda=\frac{\mu}{s_1},
 \qquad 0\leq\tau\leq\frac14.
 \label{eq:rescaled-support-variables}
\end{equation}
At \(d=8\) and \(q=0\), the homogeneity of Eq.~\eqref{eq:general-cubic}
reduces the cubic for the vector multiplicity block to
\begin{equation}
 P(\lambda,\tau)
 =\lambda^3-86\lambda^2
 +(1180+3820\tau)\lambda
 -4200-32200\tau=0.
 \label{eq:gaussian-rescaled-cubic}
\end{equation}
Thus \(\lambda\) is the largest root of this cubic.  Write
\(h(a,b):=h_0^{(8)}(a,b)=s_1\lambda(\tau)/70\).
At a differentiable supporting direction, the gradient of the support
function gives its contact point with the fidelity region:
\begin{equation}
 X=\frac{\partial h}{\partial a},
 \qquad
 Y=\frac{\partial h}{\partial b}.
 \label{eq:envelope-coordinates}
\end{equation}
Let \(\delta:=1-4\tau=(a-b)^2/s_1^2\).
Differentiating Eq.~\eqref{eq:gaussian-rescaled-cubic} implicitly, define
\begin{align}
 D_\lambda
 &:=\frac{\partial P}{\partial\lambda}
 =3\lambda^2-172\lambda+1180+3820\tau,
 \nonumber\\
 N_\lambda&:=32200-3820\lambda,
\end{align}
so that \(\lambda'(\tau)=N_\lambda/D_\lambda\).
A direct differentiation of \(h=s_1\lambda(\tau)/70\) yields
\begin{equation}
 u=\frac{2\lambda+\delta\lambda'}{70},
 \qquad
 (X-Y)^2=\frac{\delta(\lambda')^2}{70^2}.
 \label{eq:u-difference-parametric}
\end{equation}
Since \((X-Y)^2=u^2-4v\), every smooth exposed point satisfies
\begin{align}
 P(\lambda,\tau)&=0,\label{eq:elim-system-P}\\
 E_1(\lambda,\tau,u)
 &:=70uD_\lambda-2\lambda D_\lambda-\delta N_\lambda=0,
 \label{eq:elim-system-E1}\\
 E_2(\lambda,\tau,u,v)
 &:=4900(u^2-4v)D_\lambda^2-\delta N_\lambda^2=0.
 \label{eq:elim-system-E2}
\end{align}
Equation~\eqref{eq:gaussian-rescaled-cubic} is linear in \(\tau\) and gives
\begin{equation}
 \tau=-\frac{(\lambda-70)(\lambda-10)(\lambda-6)}
 {20(191\lambda-1610)}.
 \label{eq:tau-of-lambda}
\end{equation}
Substitute Eq.~\eqref{eq:tau-of-lambda} into
Eqs.~\eqref{eq:elim-system-E1} and \eqref{eq:elim-system-E2}, and let
\(\operatorname{num}\) denote the numerator after cancelling common
factors.  The polynomial resultant \(\operatorname{Res}_\lambda\)
eliminates their common root \(\lambda\), giving
\begin{multline}
 \operatorname{Res}_{\lambda}\!\left(
 \operatorname{num}E_1,\operatorname{num}E_2
 \right)=C_0\,\mathcal P(u,v),
 \qquad C_0\neq0,
 \label{eq:resultant-definition}
\end{multline}
where
\begin{multline}
\mathcal P(u,v)=1032192u^6-4066560u^5\\
{}-467712u^4v+9054652u^4-51010232u^3v\\
{}-11459032u^3+232938748u^2v^2+118873020u^2v\\
{}+6543087u^2-506334248uv^2-75277566uv\\
{}-1128886u-14829948v^3+280066507v^2\\
{}+6996458v+60835.
 \label{eq:sextic}
\end{multline}
The boundary satisfies \(\mathcal P(u,v)=0\).  The physically relevant
branch connects
\begin{equation}
 (1,1/8)
 \longrightarrow
 \left(\frac{51+\sqrt{1201}}{140},
 \frac{51+\sqrt{1201}}{140}\right)
 \longrightarrow(1/8,1).
 \label{eq:sextic-branch}
\end{equation}
An explicit parametrization selects this branch from the other real
components of the sextic.  Set \(t=\lambda\) and define
\begin{align}
 f(t)&=(t-70)(t-10)(t-6),
 &\ell(t)&=3820t-32200,\nonumber\\
 H(t)&=f'(t)\ell(t)-3820f(t),
 &J(t)&=\ell(t)+4f(t).
 \label{eq:branch-polynomials}
\end{align}
Then set
\begin{align}
 u(t)&=\frac{2tH(t)-J(t)\ell(t)}{70H(t)},\nonumber\\
 w(t)&=\frac{J(t)\ell(t)^3}{4900H(t)^2},
 &v(t)&=\frac{u(t)^2-w(t)}4.
 \label{eq:branch-rational-functions}
\end{align}
Along the largest-root branch, \(t\) decreases from \(70\) at \(\tau=0\) to
\((51+\sqrt{1201})/2\) at \(\tau=1/4\).
The half of the physical boundary with \(X\geq Y\) is
\begin{equation}
 \begin{aligned}
  \frac{51+\sqrt{1201}}2&\leq t\leq70,\\
  X&=\frac{u(t)+\sqrt{w(t)}}2,
  &Y&=\frac{u(t)-\sqrt{w(t)}}2.
 \end{aligned}
 \label{eq:physical-sextic-parametrization}
\end{equation}
and the other half is obtained by exchanging \(X\) and \(Y\).  On this
interval \(H(t)>0\) and \(w(t)\geq0\), and substitution of
\(u=X+Y\) and \(v=XY\) into Eq.~\eqref{eq:sextic} gives zero.  Thus
Eq.~\eqref{eq:sextic} is an implicit equation for the physical boundary;
its other real components are not asserted to be attainable fidelities.

\section{Uniqueness of the symmetric quadric-orbit cloner}
\label{app:unique}

At \(q=0\) and \(a=b=1\), put
\begin{equation}
 \mathcal D_d=\Nd^2-16(d-2),
 \qquad
 \mu_+=\frac{\Nd+4d+\sqrt{\mathcal D_d}}2.
\end{equation}
The remaining eigenvalues on the vector isotypic component are
\begin{equation}
 \mu_-=\frac{\Nd+4d-\sqrt{\mathcal D_d}}2,
 \qquad
 \mu_0=\Nd.
\end{equation}
The largest complementary branch is \(2(d+2)\).  For \(d>3\), \(\Nd>2(d+2)\); at \(d=3\) they are equal.  In every case \(\mu_+>\Nd\).  For \(d=3\) this is immediate.  For \(d\geq4\), one has \(\Nd-4d>0\) and
\(\mathcal D_d-(\Nd-4d)^2=8(d-1)(d-2)(d+2)>0\).
Thus \(\mu_+\) is simple on the multiplicity space and is strictly separated from every other branch.  The full top eigenspace is one copy of the vector representation,
\begin{equation}
 E_*^{(d)}=\operatorname{ran}W_{c_*^{(d)}},
 \qquad
 \dim E_*^{(d)}=d.
 \label{eq:Estar}
\end{equation}

If \(J\) is any Choi matrix attaining the optimal symmetric quadric-orbit local
fidelity, equality in the spectral variational bound forces
\(\operatorname{supp}J\subseteq E_*^{(d)}\).  Since \(W_{c_*^{(d)}}\) is an
isometry, there is an \(X\geq0\) such that
\(J=W_{c_*^{(d)}}XW_{c_*^{(d)}}^\dagger\).
For a general real multiplicity vector \(c\), direct contraction gives
\begin{align}
 \Tr_{AB}(W_cXW_c^\dagger)
 ={}&\frac{c_1^2+c_2^2}{d}\Tr(X)I
 +\frac{2c_1c_2}{d}X^T\nonumber\\
 &+\left[c_3^2+\frac{2c_3(c_1+c_2)}{d}\right]X.
 \label{eq:partial-trace-map}
\end{align}
Decompose the Hermitian matrix \(X\) according to transpose parity as
\(X=\alpha I+S+A\), where
\begin{equation}
 \alpha=\frac{\Tr X}{d},
 \qquad S^T=S,\quad\Tr S=0,
 \qquad A^T=-A.
\end{equation}
Hermiticity implies that \(S\) is real symmetric and \(A\) is purely imaginary antisymmetric.

Write \(c_*^{(d)}=(x,x,z)^T\) and \(x/z=-\kappa_d\).  The partial-trace map in Eq.~\eqref{eq:partial-trace-map} acts on the symmetric-traceless and antisymmetric sectors with eigenvalues
\begin{align}
 \lambda_S
 &=z^2+\frac{4xz+2x^2}{d}
 =\frac{z^2}{d}\left[2(\kappa_d-1)^2+d-2\right]>0,\label{eq:lambda-S}\\
 \lambda_A
 &=z^2+\frac{4xz-2x^2}{d}
 =\frac{z^2}{d}\left[d-4\kappa_d-2\kappa_d^2\right]<0.
 \label{eq:lambda-A}
\end{align}
The second inequality follows from \(\kappa_d>d/2\), as established in Sec.~\ref{sec:null-local-global}.  Trace preservation therefore forces \(S=A=0\).  On the scalar sector, \((c_*^{(d)})^TG_dc_*^{(d)}=1\) implies that the map sends \(\alpha I\) to \(\alpha I\), so trace preservation fixes \(\alpha=1\).  Hence
\(J=J_*^{(d)}=W_{c_*^{(d)}}W_{c_*^{(d)}}^\dagger\), which proves
uniqueness.  Since \(W_{c_*^{(d)}}\) is an isometry from \(\C^d\), the optimal Choi matrix is an orthogonal projector of rank \(d\).

\bibliography{fermionic_cloning_pra_revised}

\end{document}